\documentclass[journal,12pt,draftclsnofoot, onecolumn]{IEEEtran}
\usepackage{ifpdf}
\usepackage{cite}
\usepackage[pdftex]{graphicx}
\usepackage{amsmath}
\usepackage{algorithmic}
\usepackage{array}
\usepackage{caption}
\usepackage[caption=false,font=footnotesize]{subfig}

\usepackage{fixltx2e}
\usepackage{stfloats}
\usepackage{url}
\usepackage{amsthm}
\usepackage{amsmath}
\usepackage{amsfonts}
\usepackage{amssymb}
\usepackage[T1]{fontenc} 
\usepackage{amsmath}
\usepackage[cmintegrals]{newtxmath}
\usepackage{bm} 
\usepackage[all]{xy}
\usepackage{enumitem}
\usepackage{float}
\usepackage{amsmath}
\usepackage[usenames, dvipsnames]{color}
\usepackage{ifpdf}
\usepackage{cite}

\usepackage{amsmath}
\usepackage{algorithmic}
\usepackage{array}
\usepackage{mathtools}
\DeclarePairedDelimiter{\abs}{\lvert}{\rvert}

\makeatletter
\def\widebreve{\mathpalette\wide@breve}
\def\wide@breve#1#2{\sbox\z@{$#1#2$}%
\mathop{\vbox{\m@th\ialign{##\crcr
\kern0.08em\brevefill#1{0.8\wd\z@}\crcr\noalign{\nointerlineskip}%
$\hss#1#2\hss$\crcr}}}\limits}
\def\brevefill#1#2{$\m@th\sbox\tw@{$#1($}%
 \hss\resizebox{#2}{\wd\tw@}{\rotatebox[origin=c]{90}{\upshape(}}\hss$}
\makeatletter

\usepackage{url}
\usepackage{mathtools}
\usepackage{amsthm}
\usepackage{amsmath}
\usepackage{amsfonts}
\usepackage{amssymb}
\usepackage[T1]{fontenc} 
\usepackage{amsmath}
\usepackage[cmintegrals]{newtxmath}
\usepackage{bm} 
\usepackage[all]{xy}
\usepackage{enumitem}
\usepackage{float}
\usepackage{amsmath}
\usepackage[dvipsnames]{xcolor}
\usepackage{multirow}
\usepackage[T1]{fontenc}
\usepackage{mwe} 
\usepackage{times}
\usepackage{wasysym}
\usepackage{epsfig}
\usepackage{multirow}

\usepackage{tabularx}
\usepackage{graphicx}
\usepackage{color}
\usepackage{xspace}
\usepackage{thumbpdf}
\usepackage{listings}
\usepackage{verbatim}
\usepackage{outlines}
\usepackage{textgreek}
\usepackage{booktabs}
\usepackage{amsfonts}
\usepackage{colortbl}
\usepackage{multicol}
\usepackage{multirow}
\usepackage{balance}
\usepackage{makecell}
\usepackage{array}
\usepackage[normalem]{ulem}
\usepackage{mathtools,amssymb}
\usepackage{float} 
\usepackage{booktabs} 
\usepackage{wrapfig}
\usepackage{subfloat}

\usepackage{diagbox}

\theoremstyle{definition}

\newtheorem{thm}{Theorem}

\newtheorem{example}{Example}
\newtheorem{lem}{Lemma}

\newtheorem{define}{Definition}
\newcommand{\cov}{\textrm{Cov}}

\newcommand{\no}{\nonumber}
\newcommand{\Bern}{\mathop{\mathrm{Bernoulli}}}
\newcommand{\Unif}{\mathop{\mathrm{Uniform}}}
\newcommand{\Bino}{\mathop{\mathrm{Binomial}}}

\begin{document}
\title{Privacy of Dependent Users Against\\ Statistical Matching}

\author{Nazanin~Takbiri~\IEEEmembership{Student Member,~IEEE,}
	Amir~Houmansadr~\IEEEmembership{Member,~IEEE,}
	Dennis~Goeckel~\IEEEmembership{Fellow,~IEEE,}
	Hossein~Pishro-Nik~\IEEEmembership{Member,~IEEE}
	\thanks{This work was supported by National Science Foundation under grants CCF--1421957 and CNS--1739462.
		This work was presented in part in IEEE International Symposium on Information Theory (ISIT 2018)~\cite{nazanin_ISIT2018}.}
	\thanks{N. Takbiri, D. L. Goeckel, and H. Pishro-Nik are with the Department
		of Electrical and Computer Engineering, University of Massachusetts, Amherst,
		MA, 01003 USA e-mail: (ntakbiri@umass.edu; goeckel@ecs.umass.edu; pishro@engin.umass.edu).}
	\thanks{A. Houmansadr is with the College of Information and Computer Sciences, University of Massachusetts, Amherst,
		MA, 01003 USA e-mail:(amir@cs.umass.edu).}
}

\maketitle

\begin{abstract}	
Modern applications significantly enhance user experience by adapting to each user's individual condition and/or preferences. While this adaptation can greatly improve a user's experience or be essential for the application to work, the exposure of user data to the application presents a significant privacy threat to the users\textemdash even when the traces are anonymized\textemdash since the statistical matching of an anonymized trace to prior user behavior can identify a user and their habits. Because of the current and growing algorithmic and computational capabilities of adversaries, provable privacy guarantees as a function of the degree of anonymization and obfuscation of the traces are necessary. Our previous work has established the requirements on anonymization and obfuscation in the case that data traces are independent between users. However, the data traces of different users will be dependent in many applications, and an adversary can potentially exploit such. In this paper, we consider the impact of dependency between user traces on their privacy. First, we demonstrate that the adversary can readily identify the association graph of the obfuscated and anonymized version of the data, revealing which user data traces are dependent. Next, we demonstrate that the adversary can use this association graph to break user privacy with significantly shorter traces than in the case of independent users, and that obfuscating data traces independently across users is often insufficient to remedy such leakage. Finally, we discuss how users can improve privacy by employing joint obfuscation that removes or reduces the data dependency.
\end{abstract}



\section{Introduction}
\label{intro}

Many modern applications provide an enhanced user experience {\color{black} by exploiting users' characteristics}, including their past choices and present states. In particular, emerging Internet of Things (IoT) applications include smart homes, healthcare, and connected vehicles that intelligently tailor their performances to their users.
For instance, a typical connected vehicle application optimizes its route selection based on the current location of the vehicle, traffic conditions, and the users' preferences.
For such applications to be able to provide their enhanced, user-tailored performances, they need to request their clients for potentially sensitive user information such as mobility behaviors and social preferences.
Therefore, such applications trade off user privacy for enhanced utility.
Previous work~\cite{Naini2016} shows that even if users' data traces are anonymized before being provided to such applications, standard statistical matching techniques can be used to leak users' private information. Thus, privacy and security threats are a major obstacle to the wide adoption of IoT applications, as demonstrated by prior studies{\color{black}~\cite{14sadeghi2015security,lin2016iot,9dalipi2016security,11al2015security, 10harris2016security,8ukil2015privacy, iscas2019,12sivaraman2015network, nad_privacy, FTC2015,0Quest2016, 3ukil2014iot, 4Hosseinzadeh2014,iotCastle,flavio,matching}}.

The bulk of previous work assumes independence between the traces of different users.~\cite{corgeo,corPLP,cordummy,cordiff,Negar,negar_database,shirani_database} have mostly considered temporal and spatial dependency within data traces, but not cross-user dependency.
In~\cite{corgeo}, an obfuscation technique is employed to achieve privacy; however, for continuous Location-Based Services (LBS) queries, there is often strong temporal dependency in the locations. Hence,~\cite{corgeo} considers how dependency of the users' obfuscated data can impact privacy, and then employs an adaptive noise level to achieve more privacy while still maintaining an acceptable level of utility. Liu et al.~\cite{cordummy} show that the spatiotemporal dependency between neighboring location sets can ruin the privacy achieved using a dummy-based location-privacy preserving mechanism (LPPM); to solve this problem, they propose a spatiotemporal dependency-aware privacy protection that perturbs the spatiotemporal dependency between neighboring locations. Zhang et al.~\cite{corPLP} employ Protecting Location Privacy (PLP) against dependency-analysis attack in crowd sensing: the potential dependency between users' data is modeled, and the data is filtered to remove the samples that disclose the user's private data. In~\cite{cordiff}, locations of a single user are temporally dependent, and $\delta$-location set based differential privacy is proposed to achieve location privacy at every timestamp. Finally, Song et al.~\cite{correlection_timeseries} provide privacy when there is dependency within the data of a single user.
In summary, previous studies do not consider dependency between users, which is the focus of this work.
We argue that for many applications, there is dependency between the traces of different users. For example, friends tend to travel together or might meet at given places, hence introducing dependency between the traces of their location information.
Several previous works~\cite{cordiff8, cordata1, cordiff2, cordiff3, cordiff4, pufferfish_kifer, dependency_Liu} have considered cross-user dependency; however, this only has been for protecting queries on aggregated data, which is different than our application scenario.

{\color{black}We use the notion of ``perfect privacy'' and ``no privacy'', as introduced in our prior work~\cite{tifs2016, Nazanin_IT}, to evaluate the privacy of user traces.} The ``perfect privacy'' notion provides an information-theoretic guarantee on privacy in the presence of a strong adversary who has complete knowledge on users' prior data traces. {\color{black} On the opposite extreme is the notion of `no privacy''. It means there exists an algorithm for the adversary to estimate the actual data points of users with diminishing error probability.}
Through a series of work~\cite{Nazanin_IT,nazanin_ISIT2017,ciss2017,ciss2018, tifs2016,Nazanin_CISS2019}, we have derived the degree of user anonymization and data obfuscation required to obtain perfect privacy\textemdash assuming that the data traces of different users are \emph{independent} across users. Particularly, we evaluated the case of independent and identically distributed (i.i.d.) samples from a given user and the case when there is temporal dependency within the trace of a given user~\cite{Nazanin_IT,nazanin_ISIT2017} (but independent across users). In this work, we expand our study to the case where there is dependency between the data traces of users. That is, we investigate how privacy is affected by the presence of dependency between the data traces of users when anonymization and obfuscation techniques are used. {\color{black} We show that dependency significantly reduces the privacy of users. Specifically, we show that the same anonymization and obfuscation levels that could produce perfect privacy for independent users result in no privacy for dependent users. Thus, in the presence of inter-user dependency, we need to employ much stronger anonymization and obfuscation compare to the case data traces of different users are independent.}

We model dependency between user traces with an association graph, where the presence of an edge between the vertices corresponding to a pair of users indicates a non-zero dependency between their data traces. We employ standard concentration inequalities to demonstrate that the adversary can readily determine this association graph. Using this association graph and statistical data about the users, the adversary can attempt to identify users, and we demonstrate that this provides the adversary with a significant advantage versus the case when the data traces of different users are independent of one another. This suggests that, unless additional countermeasures are employed, the results of~\cite{Nazanin_IT,nazanin_ISIT2017,ciss2017,tifs2016} for independent traces are optimistic when user traces are dependent. We next consider the effectiveness of countermeasures. First, we argue that adding independent obfuscation to user data points is often ineffective in improving the privacy of (dependent) users. Next, we demonstrate that, if users with dependent traces can jointly design their obfuscation, user privacy can be significantly improved.

{\color{black}
A related but parallel approach to our study is graph alignment in which the edge set is sampled at random. Graph alignment is the problem of finding a matching between the vertices of the two graphs that matches, or aligns, many edges of the first graph with edges of the second graph. Shirani et. al.~\cite{shirani1, Shirani2018TypicalityMF} and Cullina et. al.~\cite{negar2} have done significant work on graph alignment. Although the graph alignment problem looks similar to our problem on the surface, there exist notable differences between the two. First, in Shirani et. al.'s work~\cite{shirani1, Shirani2018TypicalityMF, shirani3}, graphs are generated using a model which is sampled at random from a probability distribution, while here the association graph is deterministic, as it is based on the dependency between data traces of users. Consequently, Shirani et. al.~\cite{shirani1, Shirani2018TypicalityMF, shirani3} used a completely different approach and solution to de-anonymize users. In other words, they have not used the probability distribution of the data traces of each user to break anonymization, while here the probability distribution of the data trace of each user is a key characteristic which helps the adversary to break users' privacy. Finally, Shirani et. al.~\cite{shirani1, Shirani2018TypicalityMF} considered discrete values for the correlation between users and used them to de-anonymize the graph, while here the correlation between users have continuous values and the adversary does not have access to the exact value of them.

~\cite{negar2,negar3,negar4} considered the graph alignment for two correlated graphs, while here we assume the adversary has the association graph and tries to reconstruct it from the anonymized and obfuscated data traces. Thus, in our work, the adversary has two identical graphs and their goal is to identify all of the users based on the observed data and their statistical knowledge of users. Also, Cullina et. al.~\cite{negar2} considered fractional matching, while here the adversary can identify not only all of the users but also the data points of each user at all time with small error probability.
~\cite{gr1,gr2,gr3,gr4} studied matching of non-identical pairs of correlated Erd\"{o}s-R\'{e}nyi graphs.

Also, graph isomorphism studied in~\cite{iso1,iso2,iso3,Bollobs2001RandomG} is an instance of the matching problem where the two graphs are identical copies of one another.~\cite{Bollobs2001RandomG} studied different algorithms such as maximum degree algorithms to match two identical graphs for the case where each edge of the graph has a fixed probability of being present or absent which is in the range of $\left[\omega\left(\log n /n^{\frac{1}{5}}\right), 1-\omega\left(\log n /n^{\frac{1}{5}}\right)\right]$, where $n$ is the number of vertices in the graph. Here, the approach of our work is completely different, as the adversary uses probability distributions of users' data traces to reconstruct the association graph. After reconstruction of the association graph, the adversary uses the size of each disjoint group to identify all of the members.

In summary, although matching (alignment) between graphs can be considered as a part of our analysis, the analysis based on the users' data traces and the statistical knowledge of the adversary is a key part of this paper which distinguishes it from previous works on graph alignment.

}

The rest of this paper is organized as follows. In Section~\ref{sec:framework} we present the model and metrics considered in this work. In Sections~\ref{anon} and~\ref{obfs}, we show dependency between users' traces degrades privacy. {\color{black}In Section~\ref{appendixA}, we discuss how our methodology can be applied to a more general setting for the association graph.} In Section~\ref{perfect}, we propose a method to {\color{black}improve} privacy in the case when there exists inter-user dependency. Finally, Section~\ref{conclusion} presents the conclusions and ideas for continuing work.

\subsection{Summary of the Results}

Consider a setting with $n$ total users. As in our previous work~\cite{Nazanin_IT}, privacy depends on two parameters: (1) $m=m(n)$, the number of data points after which the pseudonyms of users are changed in the anonymization technique, i.e., smaller $m$ implies higher levels of anonymization; and (2) $a_n$, which indicates the amplitude of the obfuscation noise, i.e., larger $a_n$ implies higher levels of obfuscation.

When there are a large number of users in the setting ($n \to \infty$) and each user's dataset is governed by an i.i.d.\ process with $r$ possible values for each data point (e.g., $r$ possible locations), we obtain a no-privacy region in the $m(n)-a_n$ plane. Figure~\ref{fig:region1} shows the no-privacy region for the case when there exists inter-user dependency, and Figure~\ref{fig:region2} shows the no-privacy region when the users' traces are independent across users. There exists a larger no-privacy region in the presence of inter-user dependency; therefore, we find that dependency between users weakens their privacy.

In addition, for the case where users' datasets are governed by an irreducible and aperiodic Markov chains with $r$ states and $|E|$ edges, we obtain similar results, again showing that inter-user dependency degrades user privacy.

{\color{black} Note that for only anonymization case, an initial extension in Gaussian case with known covariance matrix is also presented in~\cite{WCNC2019}. }

\begin{figure*}[t!]
	\centering
	\subfloat[The dependent case.]{
		\includegraphics[width=0.47\columnwidth]{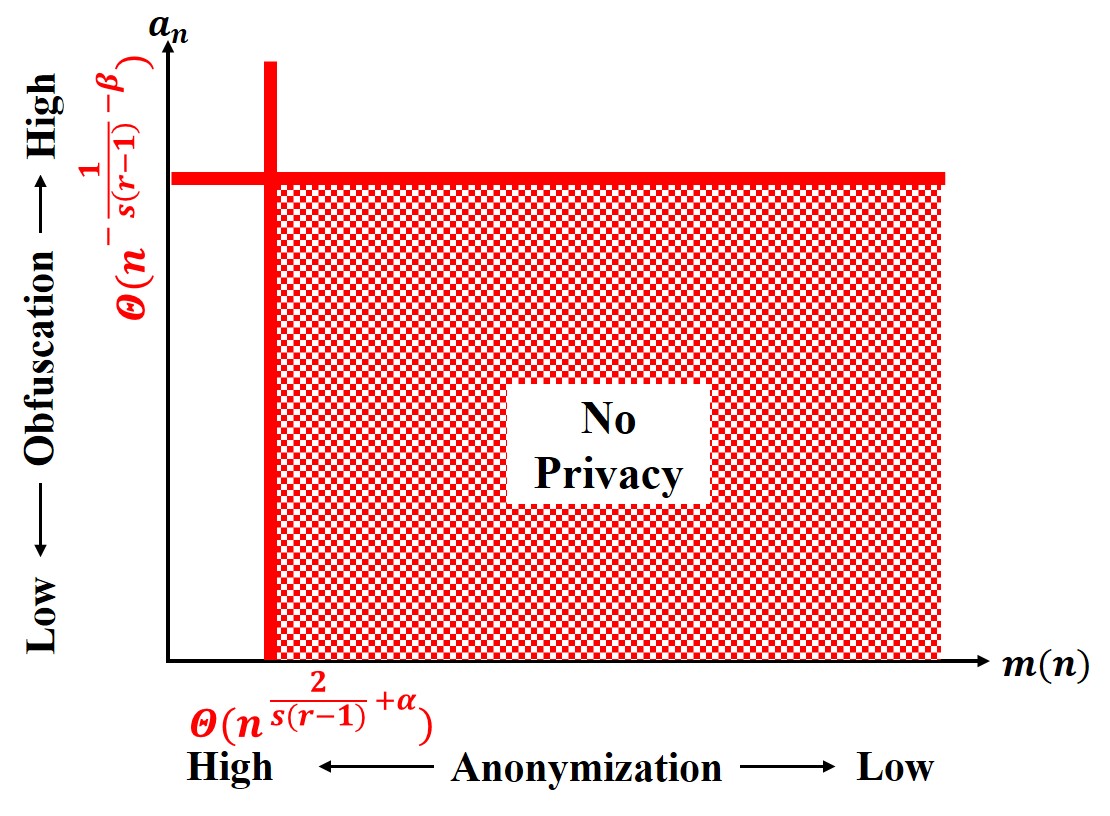}
		\label{fig:region1}
	}
	\subfloat[The independent case.]{
		\includegraphics[width=0.47\columnwidth]{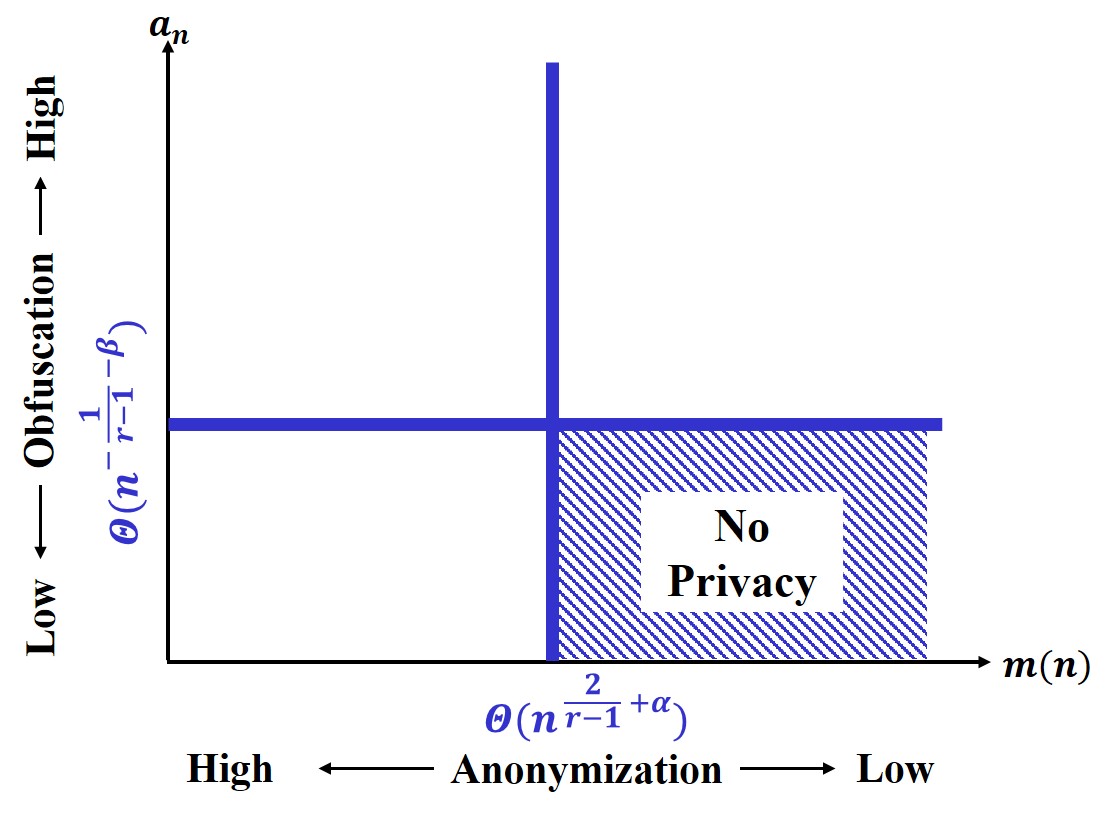}
		\label{fig:region2}
	}
	\caption{Representations of the no-privacy region in the case of dependent and independent users. Note that $m(n)$ is the number of the adversary's observations per user (degree of anonymization), and $a_n$ is the amount of noise level (degree of obfuscation). Here, the size of the group of users whose data traces are dependent is equal to $s$.}
	\label{fig:region}
\end{figure*}

\section{Framework}
	\label{sec:framework}
	Here, we employ a similar framework to~\cite{tifs2016,Nazanin_IT}. The system has $n$ users, and $X_u(k)$ is the data point of user $u$ at time $k$. Our main goal is protecting $X_u(k)$ from a strong adversary who has full knowledge of the (unique) marginal probability distribution function of the data points of each user based on previous observations or other sources. In order to achieve data privacy for users, both anonymization and obfuscation techniques can be used as shown in Figure~\ref{fig:xyz}. In Figure~\ref{fig:xyz}, $Z_u(k)$ shows the (reported) data point of user $u$ at time $k$ after applying obfuscation, and $Y_u(k)$ shows the (reported) data point of user $u$ at time $k$ after applying anonymization to $Z_u(k)$. Let $m=m(n)$ be the number of data points after which the pseudonyms of users are changed using anonymization. To break obfuscation and anonymization, the adversary tries to estimate $X_u(k)$, $k=1, 2, \cdots, m$, from $m$ observations per user by matching the sequence of observations to the known statistical characteristics of the users.
	\begin{figure}[h]
		\centering
		\includegraphics[width = 0.7\linewidth]{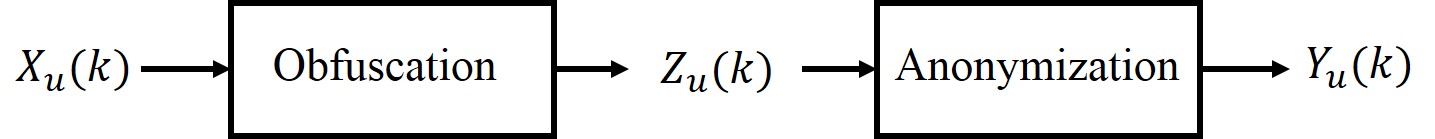}
		\caption{Applying obfuscation and anonymization techniques to the users' data points.}
		\label{fig:xyz}
	\end{figure}
	Let $\textbf{X}_u$ be the $m \times 1$ vector containing the data points of user $u$, and $\textbf{X}$ be the $m \times n$ matrix with the $u^{th}$ column equal to $\textbf{X}_u$:
	\[\textbf{X}_u = \begin{bmatrix}
	X_u(1) \\ X_u(2) \\ \vdots \\X_u(m) \end{bmatrix} , \ \ \ \textbf{X} =\left[\textbf{X}_{1}\ \ \textbf{X}_{2}\ \ \cdots \ \ \textbf{X}_{n}\right].
	\]

	\hspace{-0.18 in}\textbf{{Data Points Model:}} {\color{black} Here, we assume two different models for users' data points: in the first case, we assume the sequence of data for any individual user is modeled by i.i.d.\ which could apply directly to data that is sampled at a low rate. In addition, understanding the i.i.d.\ case can also be considered the first step toward understanding the more complicated case where there is temporal dependency. In the second case, we assume the data trace of any individual users is governed by Markov chain in which each sample of users' data is dependent over time.}
	
	We also assume users' data points can have one of $r$ possible values $(0, 1, \cdots, r-1)$. Thus, according to a user-specific probability distribution $(\textbf{p}_u)$, $X_u(k)$ is equal to a value in $\left\{0,1, \cdots, r-1 \right\}$ at any time. Note $p_u(i)$ is the probability of user $u$ having the data value $i$, so
	\[\textbf{p}_u= \begin{bmatrix}
	p_u(1) \\ p_u(2) \\ \vdots \\p_u(r-1) \end{bmatrix} , \ \ \text{for each } u\in \{1, 2, \cdots, n\}.
	\]
	We also assume $\textbf{p}_u$'s are drawn independently from some continuous density function, $f_\textbf{P}(\textbf{p}_u)$, which has support on a subset of the $(0,1)^{r-1}$ hypercube.
	Note these user-specific probability distributions, i.e., $\textbf{p}_u$'s, are known to the adversary and form the basis upon which they perform (statistical) matching.
	
	\hspace{-0.18 in}\textbf{{Association Graph:}}
	An association graph or dependency graph is an undirected graph representing dependency of the data of users with each other. Let $G(\mathcal{V},F)$ denote the association graph with set of nodes $\mathcal{V}$, $(|\mathcal{V}|=n)$, and set of edges $F$. Two vertices (users) are connected if their data sets are dependent. More specifically,
	\begin{itemize}
		\item $(u,u') \notin F$ iff $I(X_u(k); X_{u'}(k))=0,$
		\item $(u,u') \in F$ iff $I (X_u(k);X_{u'}(k))>0,$
	\end{itemize}
	where $I\left(X_u(k); X_{u'}(k)\right)$ is the mutual information between the $k^{th}$ data point of user $u$ and user $u'$\footnote{{\color{black}It is worth noting that the mechanism that determines the joint distribution of $X_u(k)$ and $X_{u'}(k)$ does not affect the results of this paper as long as the marginal densities of $X_u(k)$'s (i.e., $\textbf{p}_u$'s) are drawn independently from $f_\textbf{P}(\textbf{p}_u)$.}}.
	
	\hspace{-0.18 in}\textbf{{Obfuscation Model:}} Obfuscation perturbs the users' data points~\cite{shokri2012protecting, gruteser2003anonymous, bordenabe2014optimal}; {\color{black} in other words, the obfuscation can be viewed as passing data through a noisy channel.} Normally, in such settings, each user has only limited knowledge of the characteristics of the overall population. {\color{black}Thus, usually, a simple distributed method in which the data points of each user are reported with error with a certain probability is employed~\cite{randomizedresponse}. Note that this probability itself is generated randomly for each user.} Let $\textbf{Z}_u$ be the vector that contains
	the obfuscated version of user $u$'s data points, and ${\textbf{Z}}$ be the collection of ${\textbf{Z}}_u$ for all users,
	\[{\textbf{Z}}_u = \begin{bmatrix}
	{Z}_u(1) \\ {Z}_u(2) \\ \vdots \\ {Z}_u(m) \end{bmatrix} , \ \ \ {\textbf{Z}} =\left[ {\textbf{Z}}_{1} \ \ {\textbf{Z}}_{2}\ \ \cdots \ \ {\textbf{Z}}_{n}\right].
	\]
	
	Here, we define the \emph{asymptotic noise level} for an obfuscation technique. Loosely speaking, the asymptotic noise level of obfuscation is the highest probable percentage of data points that are corrupted. More precisely,
	for a subset of users $\bm{\mathfrak{U}}$, let $X_u(k)$ be the actual data point of user $u$ at time $k$, $u \in \bm{\mathfrak{U}}$, $k \in \{1, 2, \cdots, m\}$, and let ${Z}_u(k)$ be the obfuscated (noisy) version of $X_u(k)$. Define
	\[A_m(u) = \frac{\abs*{\left\{k:{Z}_u(k) \neq X_u(k)\right\}}}{m}.\]
	Then, the asymptotic noise level for user $u$ is defined as follows:
	\[a(u)=\inf \bigg\{ \tau \geq 0: \mathbb{P}\left(A_m(u)>\tau \right) \to 0 \text{ as } m \to \infty\bigg\}.\]
	Also, define
	\[A_m= \frac{\sum\limits_{u \in \bm{\mathfrak{U}}} \abs*{\left\{k:{Z}_u(k) \neq X_u(k)\right\}}}{m|\bm{\mathfrak{U}}|},\]
	then, the asymptotic noise level for the entire dataset is
	\[a =\inf \bigg\{ \tau \geq 0: \mathbb{P}\left(A_m>\tau \right) \to 0 \text{ as } m \to \infty\bigg\}.\]
	Note that while the above definition is given for a general case required in Section~\ref{perfect}, in practice we often use simple obfuscation techniques that employ i.i.d.\ noise sequences. Then, by the Strong Law of Large Number (SLLN),
	\[\frac{\abs*{\left\{k:{Z}_u(k) \neq X_u(k)\right\}}}{m} \xrightarrow{\text{a.s.}} \mathbb{P}\left({Z}_u(k) \neq X_u(k)\right),\]
	and for any $k$,
	\[a(u)=\mathbb{P}\left({Z}_u(k) \neq X_u(k)\right).\]
	
	\hspace{-0.18 in}\textbf{{Anonymization Model:}} In the anonymization technique, the identity of the users is perturbed~\cite{1corser2016evaluating,hoh2005protecting,freudiger2007mix, ma2009location, shokri2011quantifying2, Naini2016,soltani2017towards}. Anonymization is modeled by a random permutation $\Pi$ on the set of $n$ users. Let ${\textbf{Y}}_u$ be the vector which contains the anonymized version of ${\textbf{Z}}_u$, and ${\textbf{Y}}$ is the collection of ${\textbf{Y}}_u$ for all users, thus
	\begin{align}
	\no {\textbf{Y}} &=\textrm{Perm}\left({\textbf{Z}}_{1}, {\textbf{Z}}_{2}, \cdots, {\textbf{Z}}_{n}; \Pi \right) \\
	\nonumber &=\left[ {\textbf{Z}}_{\Pi^{-1}(1)} \ \ {\textbf{Z}}_{\Pi^{-1}(2)}\ \ \cdots \ \ {\textbf{Z}}_{\Pi^{-1}(n)}\right ] \\
	\nonumber &=\left[ {\textbf{Y}}_{1}\ \ {\textbf{Y}}_{2}\ \ \cdots \ \ {\textbf{Y}}_{n}\right], \ \
	\end{align}
	where $\textrm{Perm}( \ . \ , \Pi)$ is the permutation operation with permutation function $\Pi$. As a result, ${\textbf{Y}}_{u} = {\textbf{Z}}_{\Pi^{-1}(u)}$ and ${\textbf{Y}}_{\Pi(u)} = {\textbf{Z}}_{u}$.

	\hspace{-0.18 in}\textbf{{Adversary Model:}} We assume the adversary has full knowledge of the marginal probability distribution function of each of the users on $\{0,1,\ldots,r-1\}$. As discussed in the data points models in succeeding sections, the parameters $\mathbf{p}_u$, $u=1, 2, \cdots, n$ are drawn independently from a continuous density function, $f_{\textbf{P}}(\textbf{p}_u)$, which
	has support on a subset of a given hypercube. The density $f_{\textbf{P}}(\textbf{p}_u)$ might be unknown to the
	adversary, so all that is assumed here is that such a density exists. From the results of the paper, it will be evident that knowing or not knowing $f_{\textbf{P}}(\textbf{p}_u)$ does not change the results asymptotically.

	The adversary knows the anonymization mechanism but does not know the realization of the random permutation. The adversary also knows the obfuscation mechanism but does not know the realization of the noise parameters. And finally, the adversary knows the association graph $G(\mathcal{V},F)$, but does not necessarily know the exact nature of the dependency. That is, while the adversary knows the marginal distributions $X_u(k)$ as well as which pairs of users have strictly positive mutual information, they might not know the joint distributions or even the values of the mutual information $I(X_u(k); X_{u'}(k))$.
	
	It is critical to note that the adversary does not have any other auxiliary information or side information about users' data.

	We adopt the definitions of perfect privacy and no privacy from~\cite{tifs2016,Nazanin_IT}:
	\begin{define}
		For an algorithm for the adversary that tries to estimate the actual data point of user $u$ at time $k$, define the error probability as
		\[\mathbb{P}_e(u,k)= \mathbb{P}\left(\widetilde{X_u(k)} \neq X_u(k)\right),\]
		where $X_u(k)$ is the actual data point of user $u$ at time $k$, $\widetilde{X_u(k)}$ is the adversary's estimated data point of user $u$ at time $k$. Now, define ${\cal E}$ as the set of all possible adversary's estimators. Then, user $u$ has \emph{no privacy} at time $k$, if and only if for large enough $n$,
		\[
		\no \mathbb{P}^{*}_e(u,k)= \inf_{\cal E} {\mathbb{P}\left(\widetilde{X_u(k)} \neq X_u(k)\right)} \to 0.
		\]
		Hence, a user has \textit{no privacy }if there exists an algorithm for the adversary to estimate $X_u(k)$ with diminishing error probability as $n$ goes to infinity.
	\end{define}
	\begin{define}
		User $u$ has \emph{perfect privacy} at time $k$ if and only if
		\begin{align}
		\no \lim\limits_{n\to \infty} \mathbb{I} \left(X_u(k);{\mathbf{Y}}\right) =0,
		\end{align}
		where $\mathbb{I}\left(X_u(k);{\textbf{Y}}\right)$ denotes the mutual information between the data point of user $u$ at time $k$ and the collection of the adversary's observations for all the users.
	\end{define}

	{\color{black}
		\hspace{-0.18 in}\textbf{{Discussion $1$}:}
		The studied anonymization and obfuscation mechanisms improve user privacy at the cost of user utility. An anonymization mechanism works by frequently changing the pseudonym mappings of users to reduce the length of time series that can be exploited by statistical analysis. However, such frequent changes may also degrade the usability of the underlying application by concealing the temporal relation between a user's data points, e.g., for a dining recommendation system that makes suggestions based on the dining history of its users. On the other hand, obfuscation mechanisms work by adding noise to users' collected data, e.g., location information. The added noise may also degrade the utility of the system.
		In this work, our goal is studying the level of anonymization and obfuscation one should employ to ensure privacy with the minimum loss in utility. In other words, we derive the optimal frequency of changing user pseudonyms during anonymization, and the optimal extent of noise added by an obfuscation mechanism while guaranteeing privacy.
		
		However, like similar works in privacy~\cite{matching,1corser2016evaluating,hoh2005protecting,freudiger2007mix, ma2009location, Naini2016}, we consider the quantification of utility orthogonal to our privacy evaluations for two reasons: (1) the implications of our PPMs on utility do not impact our privacy analysis, and (2) unlike privacy, the desired level of utility is application specific.
		
	}
	{\color{black}
		\hspace{-0.18 in}\textbf{{Discussion $2$}:}
		Note that there are two kinds of dependency:
		\begin{itemize}
		\item Intra-user dependency: In this case, there is temporal and spatial dependency within data traces of one user. For example, when the data trace of a user is governed by a Markov chain model, the Markov chain characterizes temporal intra-user dependency. Thus, the adversary can benefit from this dependency and break the users' privacy. According to the results obtained in~\cite{Nazanin_IT}, if data traces of the users are governed by i.i.d. statistics, we can show users have no privacy iff $m=\Omega(n^{\frac{2}{r-1}+\alpha})$ and $a_n=O(n^{-\frac{1}{r-1}-\beta})$; however, if the data trace of users is governed by a Markov chain, we can show users have no privacy iff $m=\Omega(n^{\frac{2}{|E|-r}+\alpha})$ and $a_n=O(n^{-\frac{1}{|E|-r}-\beta})$. Most of the previous work~\cite{corgeo,corPLP,cordummy,cordiff,Negar,negar_database,shirani_database} that considers intra-user dependency assumes independence between the traces of different users, which is different from our work as described below.
	
	\item Inter-user dependency: Here, there exists dependency between the traces of different users. This is the main focus of our work. First, we demonstrate that the adversary can readily identify the association graph of the obfuscated and anonymized version of the data, revealing which user data traces are dependent. Next, we demonstrate that the adversary can use this association graph along with their statistical knowledge and the observed obfuscated and anonymized sequences to break user privacy with significantly shorter traces than in the case of independent users, and that obfuscating data traces independently across users is often insufficient to remedy such leakage.

		\end{itemize}
	}

	{\color{black}
	\hspace{-0.18 in}\textbf{{Discussion $3$}:}
	The general models of multi-user in classical information theory assume a fixed number of users and the fundamental limits of communication systems are characterized by studying the asymptotic
	limits of infinite coding blocklength~\cite{gupta2000capacity,verdu1999spectral,guo2005randomly,li2017bayesian}. However, the emerging Internet of Things enables an ever-increasing number of users to share and access information on a large scale, i.e., applications, such as ride sharing and dining recommendation applications, the number of users is large. Thus, the number of users is allowed to grow with the blocklength~\cite{chen2014many,chen2017capacity,guo2008co}, and our goal is for the asymptotic results to provide a good insight to the performance of the privacy-preserving mechanisms for these applications.
 	Moreover, both of the privacy definitions given above (perfect privacy and no privacy) are asymptotic in the number of users $(n \to \infty)$, which allows us to find clean analytical results for the fundamental limits.
	}
	\section{Impact of Dependency on Privacy Using Anonymization}
	\label{anon}
	
	In this section, we consider only anonymization and thus the obfuscation block in Figure~\ref{fig:xyz} is not present. In this case, the adversary's observation ${\textbf{Y}}$ is the anonymized version of $\textbf{X}$; thus
	\begin{align}
	\no {\textbf{Y}} &=\textrm{Perm}\left(\textbf{X}_{1}, \textbf{X}_{2}, \cdots, \textbf{X}_{n}; \Pi \right) \\
	\nonumber &=\left[ \textbf{X}_{\Pi^{-1}(1)}\ \ \textbf{X}_{\Pi^{-1}(2)}\ \ \cdots \ \ \textbf{X}_{\Pi^{-1}(n)}\right ] \\
	\nonumber &=\left[ {\textbf{Y}}_{1}\ \ {\textbf{Y}}_{2}\ \ \cdots\ \ {\textbf{Y}}_{n}\right]. \ \
	\end{align}
	
	\subsection{{\color{black}$r$-State} i.i.d.\ Model}
	\label{iidr}
	
	There is potentially dependency between the data of different users, but we assume here that the sequence of data for any individual user is i.i.d.. We also assume users' data points can have $r$ possibilities $\left(0, 1, \cdots, r-1\right)$, and $p_u(i)$ is the probability of user $u$ having the data value $i$, i.e., $p_u(i)=\mathbb{P}\left(X_u(k)=i\right)$, for $k=1, 2, \cdots, m$. We define the vectors $\textbf{p}_u$ and $\textbf{p}$ as
	\[\textbf{p}_u= \begin{bmatrix}
p_u(1) \\ p_u(2) \\ \vdots \\p_u(r-1) \end{bmatrix}, \ \ \ \textbf{p} =\left[\textbf{p}_{1}\ \ \textbf{p}_{2}\ \ \cdots\ \ \textbf{p}_{n}\right].
\]
	We also assume $\textbf{p}_u$'s are drawn independently from some continuous density function, $f_\textbf{P}(\textbf{p}_u)$, which has support on a subset of the $(0,1)^{r-1}$ hypercube. In particular, define the range of the distribution as
	\begin{align}
	\no \mathcal{R}_{\textbf{P}} &= \left\{ (x_1, x_2, \cdots, x_{r-1}) \in (0,1)^{r-1}: x_i > 0, x_1+x_2+\cdots+x_{r-1} < 1\right\},
	\end{align}
	then, we assume there are $\delta_1, \delta_2>0$ such that:
	\begin{equation}
	\begin{cases}
	\no \delta_1\leq f_{\textbf{P}}(\mathbf{p}_u) \leq \delta_2, & \textbf{p}_u \in \mathcal{R}_{\textbf{P}}.\\
	f_{\textbf{P}}(\mathbf{p}_u)=0, & \textbf{p}_u \notin \mathcal{R}_{\textbf{P}}.
	\end{cases}
	\end{equation}
	
	The adversary knows the values of $\textbf{p}_u$, $u=1, 2, \cdots, n$, and uses this knowledge to match the observed traces to the users. We will use capital letters (i.e., $\textbf{P}_u$) when we are referring to the random variable, and use lower case (i.e., $\textbf{p}_u$) to refer to the realization of $\textbf{P}_u$.
	
	A vector containing the permutation of those probabilities after anonymization is
	\begin{align}
	\no \textbf{W} &=\textrm{Perm}\left({\textbf{P}}_{1}, {\textbf{P}}_{2}, \cdots, {\textbf{P}}_{n}; \Pi \right) \\
	\nonumber &=\left[ {\textbf{P}}_{\Pi^{-1}(1)}\ \ {\textbf{P}}_{\Pi^{-1}(2)}\ \ \cdots\ \ {\textbf{P}}_{\Pi^{-1}(n)}\right ] \\
	\nonumber &=\left[\textbf{W}_{1}\ \ \textbf{W}_{2}\ \ \cdots\ \ \textbf{W}_{n}\right ], \ \
	\end{align}
	where $\textbf{W}_{u} = {{\textbf{P}}}_{\Pi^{-1}(u)}$ and $\textbf{W}_{\Pi(u)} = {{\textbf{P}}}_{u}$.
	
	%
	
	{\color{black}
		In this case, we can say:
		\begin{itemize}
			\item $(u,u') \notin F$ iff for all $i,j \in\{0, 1, \cdots, r-1\}$, $p_{uu'}(i,j)=p_u(i)p_{u'}(j),$
			\item $(u,u') \in F$ iff for at least one pair of $i,j \in\{0, 1, \cdots, r-1\}$, $p_{uu'}(i,j)\neq p_u(i)p_{u'}(j),$
		\end{itemize}
	where $p_{uu'}(i,j)=\mathbb{P}\left(X_u(k)=i, X_{u'}(k)=j\right)$, $p_{u}(i)=\mathbb{P}\left(X_u(k)=i\right)$, and $p_{u'}(j)=\mathbb{P}\left(X_{u'}(k)=j\right)$.
		Note that the adversary knows the association graph $G(\mathcal{V},F)$, but does not necessarily know the joint probability distribution for each specific $(u,u') \in F$.
		The adversary observes the anonymized version of users' data traces and combines them with their full knowledge of the marginal probability distribution of each of the users and the structure of the whole association graph to break users' privacy with arbitrarily small error probability.

		In the first step, we show that the adversary can reliably reconstruct the entire association graph for \textit{the anonymized version of the data} (i.e., the observed data traces) with relatively few observations.
		
		\begin{lem}
			\label{lem1}
			Consider a general association graph $G(\mathcal{V},F)$. If the adversary obtains $m=(\log n)^3$ anonymized observations per user, they can construct $\widetilde{G}=\widetilde{G}(\widetilde{\mathcal{V}}, \widetilde{F})$, where $\widetilde{\mathcal{V}}=\{\Pi(u):u \in \mathcal{V}\}=\mathcal{V}$, such that with high probability, for all $u, u' \in \mathcal{V}$; $ (u,u')\in F$ iff $\left(\Pi(u),\Pi(u')\right)\in \widetilde{F}$. We write this statement as $\mathbb{P}(\widetilde{G}\simeq G)\to 1$, i.e., Graph $G$ and Graph $\widetilde{G}$ are isomorphic with high probability.
		\end{lem}
	
\begin{proof}
	For $u, u' \in \{1, 2, \cdots, n\}$, we normally write $v=\Pi(u)$ and $v'=\Pi(u')$. We provide an algorithm for the adversary that with high probability obtains all edges of $F$ correctly. First, for all $v,v' \in \{1, 2, \cdots, n\}$, and all $i,j \in \{0, 1, \cdots, r-1\}$ the adversary computes
	$\widetilde{p_{vv'}(i,j)}$, $\widetilde{p_v(i)}$, and $\widetilde{p_{v'}(j)}$ as follow:
	\begin{align}
	\widetilde{p_{vv'}(i,j)}=\frac{\abs*{\left\{k: Y_v(k)=i,Y_{v'}(k)=j\right\}}}{m}=\frac{\widehat{M}_{vv'}(i,j)}{m},\ \
	\label{puu'-tilde}
	\end{align}
		\begin{align}
	\widetilde{p_{v}(i)}=\frac{\abs*{\left\{k: Y_v(k)=i\right\}}}{m}=\frac{\widehat{M}_{v}(i)}{m},\ \
	\label{pu-tilde13}
	\end{align}
	\begin{align}
	\widetilde{p_{v'}(j)}=\frac{\abs*{\left\{k: Y_{v'}(k)=j\right\}}}{m}=\frac{\widehat{M}_{v'}(j)}{m},\ \
	\label{pu-tilde12}
	\end{align}
	where
	\[\widehat{M}_{vv'}(i,j)= \abs*{\left\{k: Y_v(k)=i, Y_{v'}(k)=j\right\}}.\]
	\[\widehat{M}_{v}(i)= \abs*{\left\{k: Y_v(k)=i\right\}}.\]
	\[\widehat{M}_{v'}(j)= \abs*{\left\{k: Y_{v'}(k)=j\right\}}.\]

	After observing $m=\left(\log n\right)^3$ data points per user and computing the above expressions, the adversary constructs $\widetilde{G}$ in the following way:
	
	\begin{itemize}
		\item If $\abs*{\frac{\widehat{M}_{vv'}(i,j)}{m}-\frac{\widehat{M}_{v}(i)}{m}\frac{\widehat{M}_{v'}(j)}{m}}\leq {m^{-\frac{1}{5}}}$ for all $i,j \in\{0, 1, \cdots, r-1\}$, then $\left(v,v'\right)\notin \widetilde{F}.$
		
		\vspace{0.2 in}
		\item If $\abs*{\frac{\widehat{M}_{vv'}(i,j)}{m}-\frac{\widehat{M}_{v}(i)}{m}\frac{\widehat{M}_{v'}(j)}{m}}\geq {m^{-\frac{1}{5}}}$ for at least one pair of $i,j \in\{0, 1, \cdots, r-1\}$, then $\left(v,v'\right)\in \widetilde{F}.$
	\end{itemize}
	We show the above method yields $\mathbb{P}(\widetilde{G}\simeq G)\to 1$ as $n \to \infty$, as follows.
	Note
	\[\widehat{M}_{vv'}(i,j) \sim { \Bino } (m, w_{vv'}(i,j)),\] \[\widehat{M}_{v}(i) \sim { \Bino } (m, w_{v}(i)),\]\[\widehat{M}_{v'}(j) \sim { \Bino } (m, w_{v'}(j)),\] 
	where $w_{vv'}(i,j)=\mathbb{P}\left(Y_v(k)=i, Y_{v'}(k)=j\right)=p_{\Pi^{-1}(v)\Pi^{-1}(v')}(i,j)$, $w_{v}(i)=\mathbb{P}\left(Y_v(k)=i\right)=p_{\Pi^{-1}(v)}(i)$, and $w_{v'}(j)=\mathbb{P}\left(Y_{v'}(k)=j\right)=p_{\Pi^{-1}(v')}(j)$. 
	Now, for all $v,v' \in\{1,2, \cdots, n\}$ and all $i,j \in \{0,1,\cdots, r-1\}$, define
	\[ \mathcal{J}_{vv'}(i,j)=\left\{\abs*{\frac{M_{vv'}(i,j)}{m}-w_{vv'}(i,j)}\geq m^{-\frac{1}{4}}\right\},\]	
	then, for all $v , v' \in \{1,2, \cdots, n \}$ and all $i, j \in \{0,1,\cdots, r-1\}$, the Chernoff bound yields
	\begin{align}
	\no \mathbb{P}\left(\mathcal{J}_{vv'}(i,j)\right)&\leq 2e^{-\frac{\sqrt{m}}{3w_{vv'}(i,j)}}\leq 2e^{-\frac{\sqrt{m}}{3}}.
	\end{align}
	Similarly, for all $v,v' \in\{1,2, \cdots, n\}$ and all $i,j \in \{0,1,\cdots, r-1\}$, define
	\[\mathcal{J}_{v}(i)=\left\{\abs*{\frac{M_{v}(i)}{m}-w_{v}(i)}\geq m^{-\frac{1}{4}}\right\},\]
	\[\mathcal{J}_{v'}(j)=\left\{\abs*{\frac{M_{v'}(j)}{m}-w_{v'}(j)}\geq m^{-\frac{1}{4}}\right\},\]
	then, the Chernoff bound yields,
	\begin{align}
	\no \mathbb{P}\left(\mathcal{J}_{v}(i)\right) \leq 2e^{-\frac{\sqrt{m}}{3}},\ \
	\end{align}
	\begin{align}
	\no \mathbb{P}\left(\mathcal{J}_{v'}(j)\right) \leq 2e^{-\frac{\sqrt{m}}{3}},\ \
	\end{align}
	Now, by employing a union bound, for all $v,v' \in\{1,2, \cdots, n\}$ and all $i,j \in \{0,1,\cdots, r-1\}$, we have
	\begin{align}
	\no \mathbb{P}\left(\mathcal{J}_{vv'}(i,j) \cup \mathcal{J}_{v}(i) \cup \mathcal{J}_{v'}(j) \right)\leq 2\left(e^{-\frac{\sqrt{m}}{3}}+e^{-\frac{\sqrt{m}}{3}}+e^{-\frac{\sqrt{m}}{3}}\right)=6e^{-\frac{\sqrt{m}}{3}}.\ \
	\end{align}
	Then, by employing a union bound again,
	\begin{align}
	\no \mathbb{P}\left(\bigcup\limits_{v=1}^n \bigcup\limits_{v'=1}^n\bigcup\limits_{i=0}^{r-1} \bigcup\limits_{j=0}^{r-1}\left\{\mathcal{J}_{vv'}(i,j) \cup \mathcal{J}_{v}(i) \cup \mathcal{J}_{v'}(j)\right\}\right) &\leq \sum\limits_{v=1}^n \sum\limits_{v'=1}^n\sum\limits_{i=0}^{r-1} \sum\limits_{j=0}^{r-1} 6e^{-\frac{\sqrt{m}}{3}}\\
	\no & = 6n^2r^2e^{-\frac{\sqrt{m}}{3}}
	\\
	&= 6r^2\exp\left\{2 \log n-\frac{(\log n)^{\frac{3}{2}}}{3}\right\}\to 0,\ \
	\label{eq_1_1}
	\end{align}
	as $n \to \infty$.
Thus, (\ref{eq_1_1}) yields that with high probability, for all $v,v' \in \{1,2,\cdots, n\}$ and all $i, j \in \{0, 1,\cdots, r-1\}$, we have
	\begin{align}
	0\leq mw_{vv'}(i,j)-m^{\frac{3}{4}} \leq \widehat{M}_{vv'}(i,j) \leq mw_{vv'}(i,j)+m^{\frac{3}{4}}.\ \
	\label{M1}
	\end{align}
	\begin{align}
	0\leq mw_v(i)-m^{\frac{3}{4}} \leq \widehat{M}_v(i) \leq mw_v(i)+m^{\frac{3}{4}}.\ \
	\label{M2}
	\end{align}
	\begin{align}
	0\leq mw_{v'}(j)-m^{\frac{3}{4}} \leq \widehat{M}_{v'}(j) \leq mw_{v}(j)+m^{\frac{3}{4}}.\ \
	\label{M3}
	\end{align}
Let us define event $A_{vv'}(i,j)$ as the event that (\ref{M1}), (\ref{M2}), and (\ref{M3}) are all valid, thus, as shown in $(\ref{eq_1_1})$, we have
\begin{align}
\mathbb{P}\left(\bigcap\limits_{v=1}^n \bigcap\limits_{v'=1}^n\bigcap\limits_{i=0}^{r-1} \bigcap\limits_{j=0}^{r-1}\left\{A_{vv'}(i,j)\right\}\right) \to 1,\ \
\label{A}
\end{align}
as $n \to \infty$. Now, if $A_{vv'}(i,j)$ is true for some $v,v' \in \{1,2,\cdots, n\}$ and some $i, j \in \{0, 1,\cdots, r-1\}$, we have
\begin{align}
	\no \frac{\widehat{M}_{vv'}(i,j)}{m}-\frac{\widehat{M}_{v}(i)}{m}\frac{\widehat{M}_{v'}(j)}{m}&\leq\frac{mw_{vv'}(i,j)+m^{\frac{3}{4}}}{m}-\frac{mw_v(i)-m^{\frac{3}{4}}}{m}\frac{mw_{v'}(i)-m^{\frac{3}{4}}}{m}\\
	\no&= w_{vv'}(i,j)-w_v(i)w_{v'}(j)+m^{-\frac{1}{4}}+ (w_v(i)+w_{v'}(j))m^{-\frac{1}{4}}-m^{-\frac{1}{2}}\\
	&\leq w_{vv'}(i,j)-w_v(i)w_{v'}(j)+m^{-\frac{1}{4}}+ (w_v(i)+w_{v'}(j))m^{-\frac{1}{4}}+m^{-\frac{1}{2}}.\ \
	\label{eq1}
	\end{align}
	Similarly,
	\begin{align}
	\no \frac{\widehat{M}_{vv'}(i,j)}{m}-\frac{\widehat{M}_{v}(i)}{m}\frac{\widehat{M}_{v'}(j)}{m}&\geq \frac{mw_{vv'}(i,j)-m^{\frac{3}{4}}}{m}-\frac{mw_v(i)+m^{\frac{3}{4}}}{m}\frac{mw_{v'}(i)+m^{\frac{3}{4}}}{m}\\
	&= w_{vv'}(i,j)-w_v(i)w_{v'}(j)-m^{-\frac{1}{4}}- (w_v(i)+w_{v'}(j))m^{-\frac{1}{4}}-m^{-\frac{1}{2}}.\ \
	\label{eq2}
	\end{align}
	Thus, by using (\ref{eq1}) and (\ref{eq2}), we have
	\begin{align}
	\abs*{\left(\frac{\widehat{M}_{vv'}(i,j)}{m}-\frac{\widehat{M}_{v}(i)}{m}\frac{\widehat{M}_{v'}(j)}{m}\right)-\left(w_{vv'}(i,j)-w_v(i)w_{v'}(j)\right)}\leq (1+w_v(i)+w_{v'}(j))m^{-\frac{1}{4}}+m^{-\frac{1}{2}}.\ \
	\label{eq3}
	\end{align}
	
	Let us define event $B_{vv'}(i,j)$ as the event that (\ref{eq3}) is valid for $v$, $v'$, $i$, and $j$. We have shown, for all $v,v' \in \{1,2,\cdots, n\}$ and all $i, j \in \{0, 1,\cdots, r-1\}$, $A_{vv'}(i,j) \subseteq B_{vv'}(i,j)$, thus
$$\left\{\bigcap\limits_{v=1}^n \bigcap\limits_{v'=1}^n\bigcap\limits_{i=0}^{r-1} \bigcap\limits_{j=0}^{r-1}\left\{A_{vv'}(i,j)\right\} \right\} \subseteq \left\{\bigcap\limits_{v=1}^n \bigcap\limits_{v'=1}^n\bigcap\limits_{i=0}^{r-1} \bigcap\limits_{j=0}^{r-1}\left\{B_{vv'}(i,j)\right\} \right\}, $$
and as a result,
\begin{align}
\no \mathbb{P}\left(\bigcap\limits_{v=1}^n \bigcap\limits_{v'=1}^n\bigcap\limits_{i=0}^{r-1} \bigcap\limits_{j=0}^{r-1}\left\{B_{vv'}(i,j)\right\}\right) \geq \mathbb{P}\left(\bigcap\limits_{v=1}^n \bigcap\limits_{v'=1}^n\bigcap\limits_{i=0}^{r-1} \bigcap\limits_{j=0}^{r-1}\left\{A_{vv'}(i,j)\right\}\right).\ \
\end{align}
Thus, by using (\ref{A}), we have 
\begin{align}
\no \mathbb{P}\left(\bigcap\limits_{v=1}^n \bigcap\limits_{v'=1}^n\bigcap\limits_{i=0}^{r-1} \bigcap\limits_{j=0}^{r-1}\left\{B_{vv'}(i,j)\right\}\right) \to 1,\ \
\end{align}
as $n \to \infty$. Hence, with high probability, for all $v,v' \in \{1,2,\cdots, n\}$ and all $i, j \in \{0, 1,\cdots, r-1\}$, we have
\begin{align}
	\abs*{\left(\frac{\widehat{M}_{vv'}(i,j)}{m}-\frac{\widehat{M}_{v}(i)}{m}\frac{\widehat{M}_{v'}(j)}{m}\right)-\left(w_{vv'}(i,j)-w_v(i)w_{v'}(j)\right)}\leq (1+w_v(i)+w_{v'}(j))m^{-\frac{1}{4}}+m^{-\frac{1}{2}}.\ \
	\label{need1}
\end{align}
	

	Now, if $(u,u') \notin F$, then for all $i,j \in\{0, 1, \cdots, r-1\}$, we have $p_{uu'}(i,j)-p_u(i)p_{u'}(j)=0$, and as a result, $w_{vv'}(i,j)-w_v(i)w_{v'}(j)=0$. Thus, by using (\ref{need1}), we have
	\begin{align}
	\no \abs*{\frac{\widehat{M}_{vv'}(i,j)}{m}-\frac{\widehat{M}_{v}(i)}{m}\frac{\widehat{M}_{v'}(j)}{m}} \leq (1+w_v(i)+w_{v'}(j))m^{-\frac{1}{4}}+m^{-\frac{1}{2}},
	\end{align}
and as a result, for large enough $m$,
	\begin{align}
	\no \abs*{\frac{\widehat{M}_{vv'}(i,j)}{m}-\frac{\widehat{M}_{v}(i)}{m}\frac{\widehat{M}_{v'}(j)}{m}} \leq m^{-\frac{1}{5}}.
	\end{align}
Thus, we can conclude, $\left(v,v'\right) \notin \widetilde{F}$, and in other words, $\left(\Pi(u),\Pi(u')\right) \notin \widetilde{F}$. This is true with high probability, for all $u, u' \in \{1, 2, \cdots, n\}$ where $(u,u') \notin F.$ 

Similarly, if $(u,u') \in F$, there exists at least one pair of $i,j \in\{0, 1, \cdots, r-1\}$ with $p_{uu'}(i,j)- p_u(i)p_{u'}(j) \geq \epsilon-m^{-\frac{1}{4}}$ for a fixed value of $\epsilon$. Thus, there exists at least one pair of $i,j \in\{0, 1, \cdots, r-1\}$ with $w_{vv'}(i,j)- w_v(i)w_{v'}(j) \geq \epsilon-m^{-\frac{1}{4}}$. As a result, by using (\ref{need1}), for large enough $m$, we have
\begin{align}
\no \abs*{\frac{\widehat{M}_{vv'}(i,j)}{m}-\frac{\widehat{M}_{v}(i)}{m}\frac{\widehat{M}_{v'}(j)}{m}} \geq m^{-\frac{1}{5}}.
\end{align}
Thus, we can conclude, $\left(v,v'\right) \in \widetilde{F}$, and in other words, $\left(\Pi(u),\Pi(u')\right) \in \widetilde{F}$. Again, this is true with high probability, for all $u, u' \in \{1, 2, \cdots, n\}$ where $(u,u') \in F.$ 
	
Now, we can conclude, for large enough $n$, we have $\mathbb{P}\left(\widetilde{G}\simeq G\right)\to 1$, so the adversary can reconstruct the association graph of the anonymized version of the data with an arbitrarily small error probability. Note that reconstruction of the association graph does not require the adversary's knowledge about user statistics (i.e., the values of $\textbf{p}_u$'s).\end{proof}
}

%
The structure of the association graph $(G)$ can leak a lot of information. 
{\color{black} For the rest of this section, we consider a graph structure shown in Figure~\ref{fig:graph}. In this structure, $G_l$, the subgraph consisting of the users the adversary wants to de-anonymize, has $s_l$ vertices and is disjoint from the reminder of the association graph. So, we can write $G_l(\mathcal{V}_l, F_l)$, where $|\mathcal{V}_l|=s_l$. Note that we assume $s_l$ is finite. In particular, the subgraph $G_l$ can be thought of as a group of ``friends'' or ``associates'' such that their data sets are dependent. In Section~\ref{appendixA}, we discuss how our methodology can be applied to the settings where the subgraph $G_l$ is not disjoint from the reminder of the graph $(G')$~\cite{Com1, Com2, Com3,Com4, Com5, Com6,Com7,Com8, Com_Negar, Com_Negar2}.}


\begin{figure}
	\centering
	\includegraphics[width=.7\linewidth]{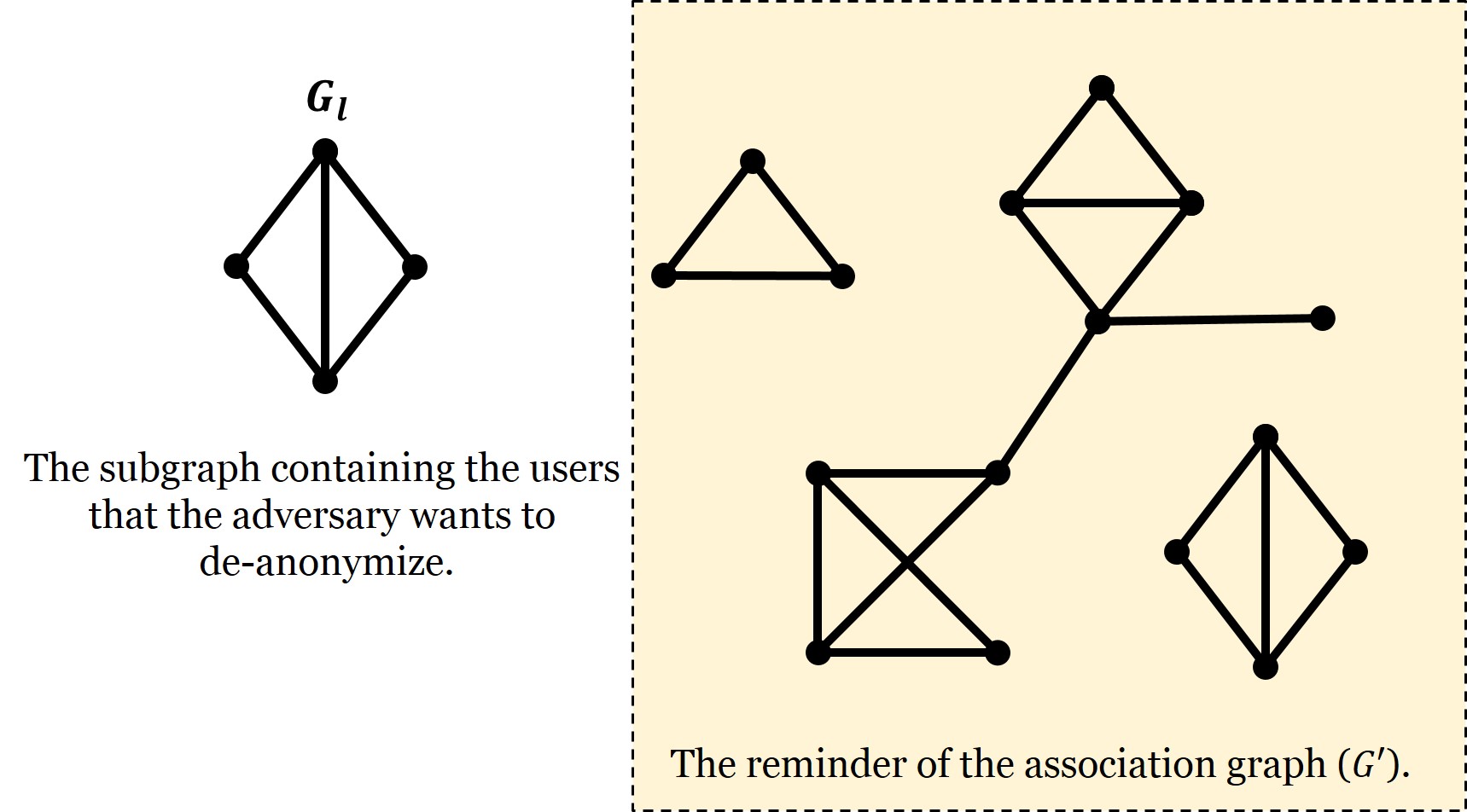}
	\centering
	\caption{The structure of the association graph ($G$):
	Group $l$ with $s_l$ vertices is disjoint from the reminder of the association graph ($G'$).}
	\label{fig:graph}
\end{figure}

The following theorem states that if the number of observations per user $(m)$ is significantly larger than $n^{\frac{2}{s(r-1)}+\alpha} $ in the $r$-state model, where $s$ is the size of a group, then the adversary can successfully de-anonymize the users in that group.

\begin{thm}\label{r_state_thm}
	For the above $r$-state model, if $\textbf{Y}$ is the anonymized version of $\textbf{X}$ as defined above, the size of the group including user $1$ is $s$, and
	\begin{itemize}
		\item {\color{black}$m=\Omega \left(n^{\frac{2}{s(r-1)}+\alpha}\right)$}, for any $\alpha > 0$;
	\end{itemize}
	then, user $1$ has no privacy at time $k$.
\end{thm}

\hspace{-0.18 in}\textbf{{Discussion $4$}:}
It is insightful to compare this result to~\cite[Theorem 2]{tifs2016}, where it is stated that if the users are not dependent, then all users have perfect privacy as long as the number of adversary's observations per user $(m)$ is smaller than {\color{black}$O(n^{\frac{2}{r-1}-\alpha})$}. Here, Theorem~\ref{r_state_thm} states that with much smaller $m$, the adversary can de-anonymize the users. Therefore, we see that dependency can significantly reduce the privacy of users.

{\color{black}
\noindent \textbf{Proof of Theorem~\ref{r_state_thm}:}
\begin{proof}
	As shown in Figure~\ref{fig:algorithm}, the proof of Theorem~\ref{r_state_thm} consists of three parts:
	\begin{itemize}
		\item \textbf{First Step: }Showing the adversary can reconstruct the association graph of the anonymized version of the data with an arbitrarily small error probability (as shown in Figure~\ref{fig:lem1}).
		\item \textbf{Second Step:} Showing the adversary can uniquely identify Group $1$ with an arbitrarily small error probability (as shown in Figure~\ref{fig:lem2}).
		\item \textbf{Third Step:} Showing the adversary can individually identify all the members within Group $1$ with an arbitrarily small error probability (as shown in Figure~\ref{fig:lem3}).
	\end{itemize}
	
	{\color{black}
	\begin{figure}[htp]
		\centering
		\subfloat[First step: Reconstruction of the association graph from the observed data.]{%
			\includegraphics[clip,width=0.8\columnwidth]{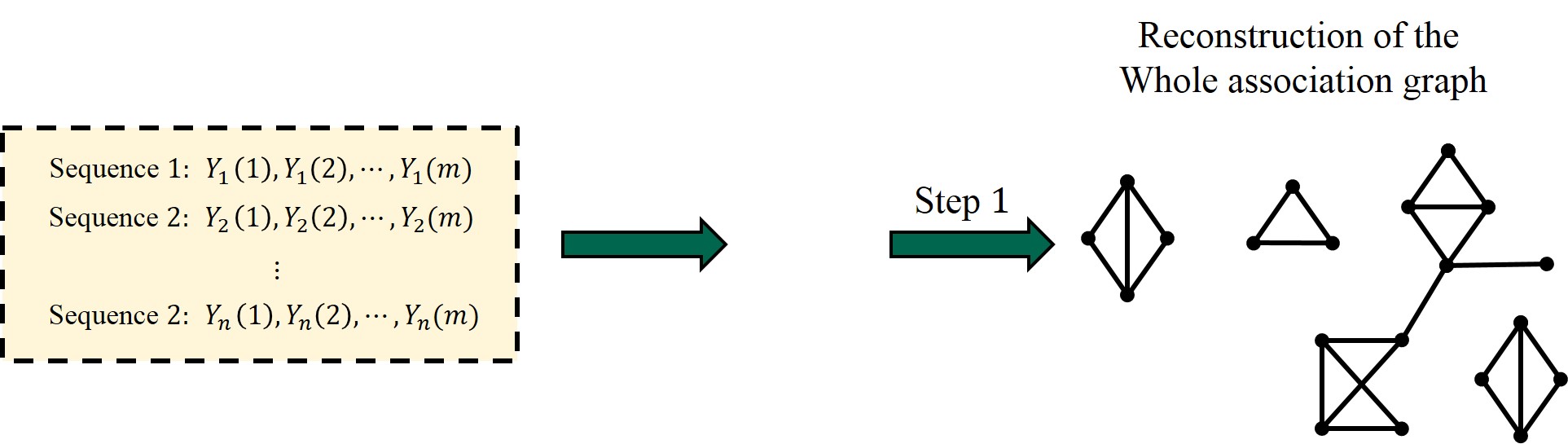}%
			\label{fig:lem1}
		}
		
		\subfloat[Second step: Identifying Group $1$ among all of the groups after association graph is reconstructed.]{%
			\includegraphics[clip,width=0.8\columnwidth]{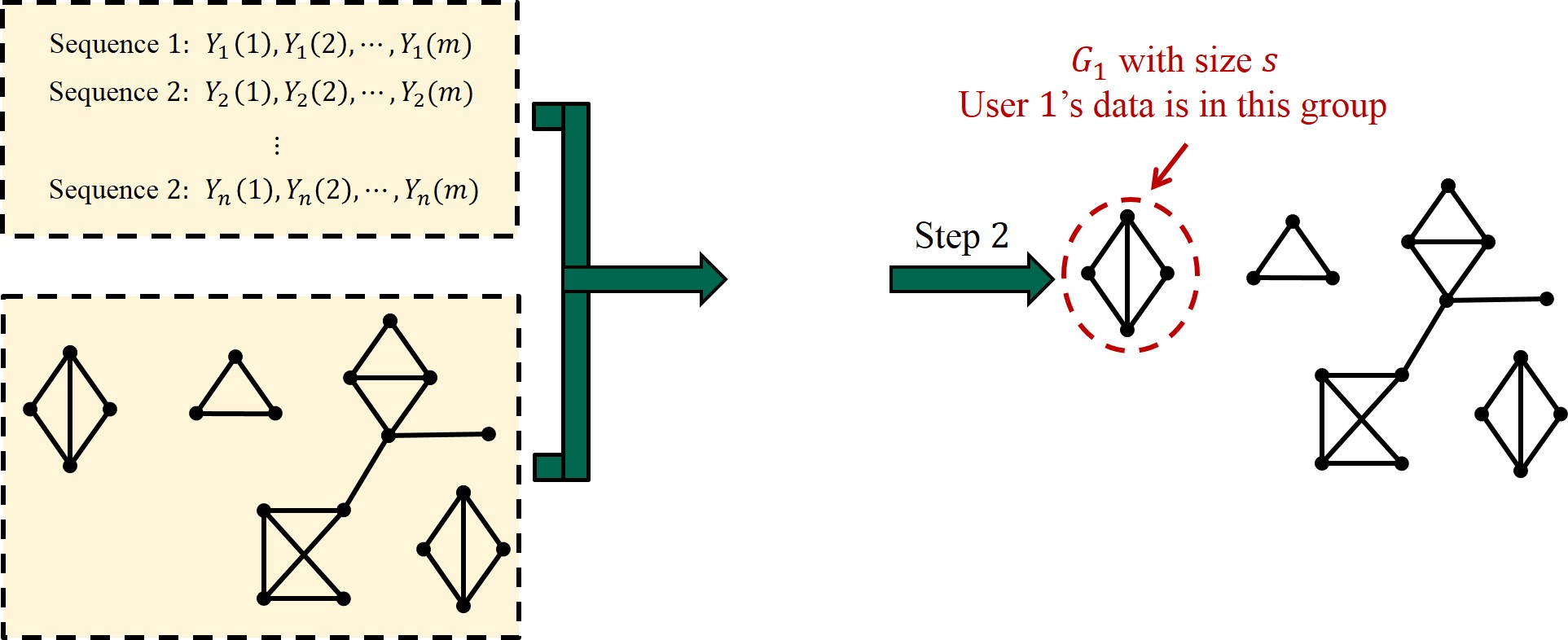}%
			\label{fig:lem2}
		}
	
		\subfloat[Third step: Identifying user $1$ among all of the members of Group $1$ after Group $1$ is uniquely identified.]{%
			\includegraphics[clip,width=0.8\columnwidth]{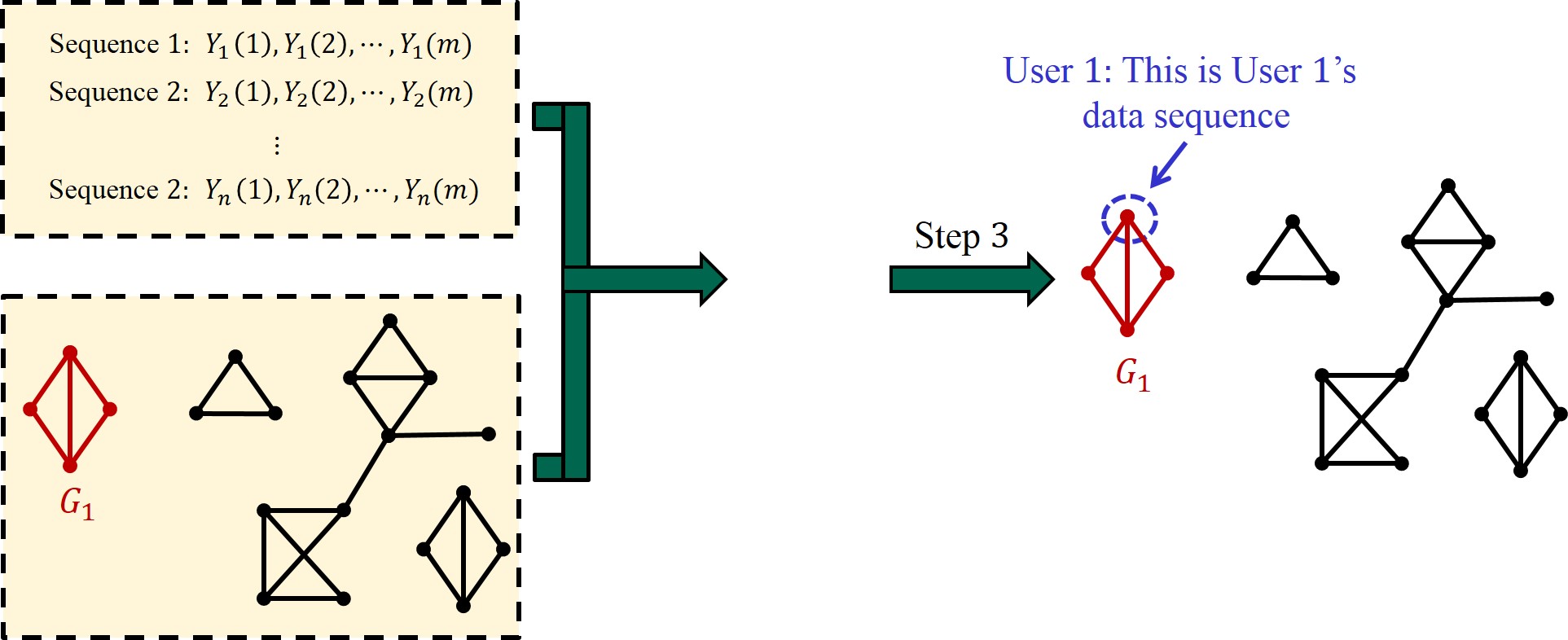}%
			\label{fig:lem3}
		}
		
		\caption{The algorithm of the adversary to estimate data points of user $1$ with vanishing error probability.}
		\label{fig:algorithm}
		
	\end{figure}
}

	The first part of the proof exploits the fact that the adversary can readily reconstruct the association graph of the anonymized data in $m=\left(\log n\right)^3$. It is the second and third parts that give rise to the condition $m=\Omega\left(n^{\frac{2}{s(r-1)}+\alpha}\right)$, and it is in the second part where we see the mechanism for the speed-up of the adversary's algorithm relative to the case where user traces are independent. In particular, due to the dependence between users breaking them into groups, the key search for the adversary now involves finding a set of users corresponding to a length-$s$ vector of probabilities rather than searching for a single user associated with a given probability.

	\hspace{-0.18 in}\textbf{First step: Reconstruction of the association graph:} In this step, we use Lemma~\ref{lem1}. More specifically, since $n^{\frac{2}{s(r-1)}}>(\log n)^3$ for large enough $n$, we can use Lemma~\ref{lem1} to conclude that the adversary can reconstruct the association graph with arbitrarily small error probability.

	\hspace{-0.18 in}\textbf{Second step: Identifying Group $1$ among all of the groups:} Now, assume the size of Group $1$ is $s$. Without loss of generality, suppose the members of Group $1$ are users $\{1,2, \cdots, s\}$. Note that there are at most $\frac{n}{s}$ isolated groups of size $s$ in the association graph. We call these Groups $1, 2, \cdots, \frac{n}{s}$. The adversary needs to first identify Group $1$ among all of these groups.

	First, for all $u\in\{1, 2, \cdots, n\}$ and all $i \in \{1, 2, \cdots, r-1\}$, the adversary computes $\widetilde{p_u(i)}$ as:
	\begin{align}
\widetilde{p_{u}(i)}=\frac{\abs*{\left\{k: Y_u(k)=i\right\}}}{m}=\frac{\widehat{M}_u(i)}{m},\ \
	\label{pu-tilde}
	\end{align}
and as a result,	
	\begin{align}
\widetilde{p_{\Pi(u)}(i)}=\frac{\abs*{\left\{k: X_u(k)=i\right\}}}{m}=\frac{M_{u}(i)}{m},\ \
	\label{pu-tilde2}
	\end{align}
	where $\widehat{M}_{u}(i)= \abs*{\left\{k: Y_u(k)=i\right\}}$ and $M_{u}(i)= \abs*{\left\{k: X_u(k)=i\right\}}.$
	 Let $\textbf{p}_u$ be the collection of ${{p}}_u(i)$ and $\widetilde{\textbf{p}_{\Pi(u)}}$ be the collection of ${\widetilde{{p}_{\Pi(u)}(i)}}$ for all $i \in \{1, 2, \cdots, r-1\}$:
	\[\widetilde{{\textbf{p}}_u} = \begin{bmatrix}
	\widetilde{p_u(1)} \\ \widetilde{p_u(2)}\\ \vdots \\ \widetilde{p_u(r-1)} \end{bmatrix} , \ \ \ \widetilde{{\textbf{p}}_{\Pi(u)}} = \begin{bmatrix}
	\widetilde{p_{\Pi(u)}(1)} \\ \widetilde{p_{\Pi(u)}(2)}\\ \vdots \\ \widetilde{p_{\Pi(u)}(r-1)} \end{bmatrix}.
	\]

	Now, define $\Sigma_s$ as the set of all permutations on $s$ elements; for $\sigma \in \Sigma_s$, $\sigma:\{1, 2, \cdots, s\}\to \{1, 2, \cdots, s\}$ is a one-to-one mapping.

	First, we provide the definition of a distance measure $D\left(\mathbf{\Phi}, \mathbf{\Psi}\right)$ for vectors
	\[\mathbf{\Phi}=\left[\mathbf{\Phi}_1 \ \ \mathbf{\Phi}_2 \ \ \cdots \ \ \mathbf{\Phi}_s\right],\]
	\[\mathbf{\Psi}=\left[\mathbf{\Psi}_1 \ \ \mathbf{\Psi}_2 \ \ \cdots \ \ \mathbf{\Psi}_s\right],\]
	where $\mathbf{\Phi}_u \in \mathbb{R}^{r-1}$ and $\mathbf{\Psi}_u \in \mathbb{R}^{r-1}$. Define
	\begin{align}
	\no D\left(\mathbf{\Phi},\mathbf{\Psi}\right)=\min\limits_{\sigma \in \Sigma_s}\left \{\max\big \{||\mathbf{\Phi}_1-\mathbf{\Psi}_{\sigma(1)}||_{\infty}, ||\mathbf{\Phi}_2-\mathbf{\Psi}_{\sigma(2)}||_{\infty}, \cdots, ||\mathbf{\Phi}_s-\mathbf{\Psi}_{\sigma(s)}||_{\infty}\big \}\right\},
	\end{align}
	where for all $u \in \{1, 2, \cdots, s\}$,
	{\begin{align}
		\nonumber ||\mathbf{\Phi}_u-\mathbf{\Psi}_{\sigma(u)}||_{\infty}=\max\left\{|\Phi_u(i)-\Psi_{\sigma(u)}(i)|:i = 1, 2, \cdots, r-1\right\}. \
		\end{align}
	}
	
	Here, let $\textbf{P}^{(l)}$ be a vector which contains probability distributions of users belonging to Group $l$, and $\widetilde{\textbf{P}_{\Pi}^{(l)}}$ be a vector which contains the estimate of the adversary about the probability distribution of users belong to Group $l$. For example, for Group $1$, we have
	\[\textbf{P}^{(1)}=\left[\textbf{p}_1 \ \ \textbf{p}_2 \ \ \cdots \ \ \textbf{p}_s\right],\]
	and
	\[\widetilde{\textbf{P}_{\Pi}^{(1)}}=\left[\widetilde{\textbf{p}_{\Pi(1)}} \ \ \widetilde{\textbf{p}_{\Pi(2)}} \ \ \cdots \ \ \widetilde{\textbf{p}_{\Pi(s)}}\right].\] 		
	Now, we claim for $m=cn^{\frac{2}{s(r-1)}+\alpha}$ and large enough $n$,
	\begin{itemize}
		\item $\mathbb{P}\left(D\left(\textbf{P}^{(1)}, \widetilde{\textbf{P}_{\Pi}^{(1)}}\right)\leq \Delta_n\right) \to 1,$
		\item $\mathbb{P}\left(\bigcup\limits_{l=2}^{\frac{n}{s}} \left\{ D\left(\textbf{P}^{(1)},\widetilde{\textbf{P}_{\Pi}^{(l)}}\right)\leq \Delta_n \right\}\right) \to 0$ ,
	\end{itemize}
	where $\Delta_n={n^{-\frac{1}{s(r-1)}-\frac{\alpha}{4}}}$.

	Define the hypercubes of $\mathcal{F}^{(n)}$ and $\mathcal{H}^{(n)}$ as
	\begin{align}
	\no \mathcal{F}^{(n)}= & \left\{(\textbf{x}_1, \textbf{x}_2, \cdots ,\textbf{x}_{s}) \in \left(\mathbb{R}^{(r-1)}\right)^s:
	\max\limits_u\{|\textbf{x}_u-\textbf{p}_u|\}\leq \Delta_n, u=1, 2, \cdots ,s\right\},
	\end{align}
	\begin{align}
	\no \mathcal{H}^{(n)}= &\left\{(\textbf{x}_1, \textbf{x}_2, \cdots ,\textbf{x}_{s}) \in \left(\mathbb{R}^{(r-1)}\right)^s:\max\limits_u\{|\textbf{x}_u-\textbf{p}_u|\} \leq 2\Delta_n, u=1, 2, \cdots ,s\right\},
	\end{align}
	Figure~\ref{setBC} shows sets $\mathcal{F}^{(n)}$ and $\mathcal{H}^{(n)}$ in the case $r=s=2$.

	\begin{figure}
		\centering
		\includegraphics[width=.3\linewidth]{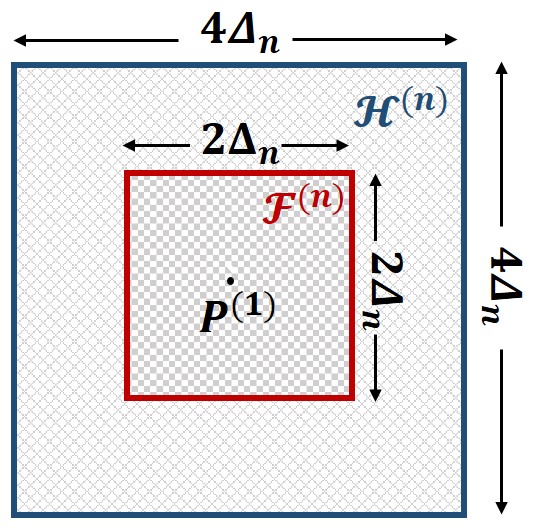}
		\centering
		\caption{$\textbf{P}^{(1)}$, sets $\mathcal{F}^{(n)}$ and $\mathcal{H}^{(n)}$ for the case $r=s=2$.}
		\label{setBC}
	\end{figure}
	
	First, we prove $\widetilde{\textbf{P}_{\Pi}^{(1)}}$ is in set $\mathcal{F}^{(n)}$, thus, $D\left(\textbf{P}^{(1)}, \widetilde{\textbf{P}_{\Pi}^{(1)}}\right)\leq \Delta_n$. Note that for all $u\in\{1, 2, \cdots, n\}$ and all $i \in \{1, 2, \cdots, r-1\}$,
	a Chernoff bound yields
	\begin{align}
	\no \mathbb{P}\left(\bigg{|}\frac{M_{u}(i)}{m}-p_{u}(i)\bigg{|}\geq \Delta_n\right)&\leq 2e^{-\frac{m\Delta_n^2}{3p_u}}\\
	\no &=2e^{-\left(cn^{\frac{2}{s(r-1)}+\alpha }\right)\left(\frac{1}{n^{\frac{1}{s(r-1)}+\frac{\alpha}{4}}}\right)^2\left(\frac{1}{3p_u}\right)}\\
	&\leq 2e^{-\frac{c}{3}n^{\frac{\alpha}{2}}}.\ \
	\label{chern}
	\end{align}
	Thus, for all $u \in \text{Group }1$ and all $i \in\{1, 2, \cdots, r-1\}$, (\ref{chern}) and the union bound yield
	\begin{align}
	\no \mathbb{P}\left(D\left(\textbf{P}^{(1)},\widetilde{\textbf{P}_{\Pi}^{(1)}}\right)\geq \Delta_n\right) & \leq \sum\limits_{u=1}^{s} \sum\limits_{i=1}^{r-1} \mathbb{P}\left(\bigg{|}\frac{M_{u}(i)}{m}-p_{u}(i)\bigg{|}\geq \Delta_n\right) \\
	\no &\leq 2s(r-1)e^{-\frac{c}{3}n^{\frac{\alpha}{2}}} \to 0,\ \
	\end{align}
	as $n \to \infty.$ As a result, $D\left(\textbf{P}^{(1)}, \widetilde{\textbf{P}_{\Pi}^{(1)}}\right)\leq \Delta_n$ with high probability.

	In the next step, we prove $\mathbb{P}\left(\bigcup\limits_{l=2}^{\frac{n}{s}} \left\{ D\left(\textbf{P}^{(1)},\widetilde{\textbf{P}_{\Pi}^{(l)}}\right)\leq \Delta_n \right\}\right) \to 0$. Note that for all groups other than Group $1$, we have
	\[(4 \Delta_n)^{s(r-1)}\delta_1\leq \mathbb{P}\left( \textbf{P}^{(l)} \in \mathcal{H}^{(n)}\right) \leq (4 \Delta_n)^{s(r-1)}\delta_2,\]
	and as a result,
	\begin{align}
	\no \mathbb{P}\left(\textbf{P}^{(l)} \in \mathcal{H}^{(n)}\right) &\leq \delta_2 (4\Delta_n)^{s(r-1)}\\
	\no &=\delta_2 4^{s(r-1)}\frac{1}{n^{1+\frac{\alpha}{4}s(r-1)}}.\ \
	\end{align}
	Similarly, for any $\sigma \in \Sigma_s$,
	\begin{align}
	\no \mathbb{P}\left(\textbf{P}^{(l)}_\sigma \in \mathcal{H}^{(n)}\right) &\leq \delta_2 (4\Delta_n)^{s(r-1)}\\
	\no &=\delta_2 4^{s(r-1)}\frac{1}{n^{1+\frac{\alpha}{4}s(r-1)}},\ \
	\end{align}
	and since $|\Sigma_s|=s!$, by a union bound,
	\begin{align}
	\no \mathbb{P}\left(\bigcup\limits_{l=2}^{\frac{n}{s}} \left\{ \bigcup\limits_{\sigma \in \Sigma_s} \left\{ \textbf{P}^{(l)}_\sigma \in \mathcal{H}^{(n)}\right\} \right\}\right) &\leq \sum\limits_{l=2}^{\frac{n}{s}} \sum\limits_{\sigma \in \Sigma_s} \mathbb{P}\left(\textbf{P}^{(l)}_\sigma \in \mathcal{H}^{(n)} \right)\\
	\no &\leq \frac{n}{s}s!\delta_2 4^{s(r-1)}\frac{1}{n^{1+\frac{\alpha}{4}s(r-1)}}\\
	\no &=(s-1)!4^{s(r-1)}\delta_2{n^{-\frac{\alpha}{4}s(r-1)}}\to 0,\ \
	\end{align}
	as $n \to \infty.$
	Thus, all $\textbf{P}^{(l)}$'s are outside of $\mathcal{H}^{(n)}$ with high probability.
	
	Now, given the fact that all $\textbf{P}^{(l)}$'s are outside of $\mathcal{H}^{(n)}$, we prove $\mathbb{P}\left(\bigcup\limits_{l=2}^{\frac{n}{s}} \left\{ D\left(\textbf{P}^{(1)},\widetilde{\textbf{P}_{\Pi}^{(l)}}\right)\leq \Delta_n \right\}\right) \to 0$. We show that $\widetilde{\textbf{P}_{\Pi}^{(l)}}$'s are close to $\textbf{P}^{(l)}$'s, and as a result, they will be outside of $\mathcal{F}^{(n)}$. 
For all $u \in \text{Group }l$ and all $i \in\{1, 2, \cdots, r-1\}$, (\ref{chern}) and the union bound yield,
	\begin{align}
	\no \mathbb{P}\left(D\left(\textbf{P}^{(1)},\widetilde{\textbf{P}_{\Pi}^{(l)}}\right)\leq \Delta_n\right)&=\mathbb{P}\left(D\left(\textbf{P}^{(l)},\widetilde{\textbf{P}_{\Pi}^{(l)}}\right)\geq \Delta_n\right)\\
	\no & \leq \sum\limits_{u=1}^{s} \sum\limits_{i=1}^{r-1} \mathbb{P}\left(\bigg{|}\frac{M_{u}(i)}{m}-p_{u}(i)\bigg{|}\geq \Delta_n\right) \\
	\no &\leq 2s(r-1)e^{-\frac{c}{3}n^{\frac{\alpha}{2}}}. \ \
	\end{align}
	Now by using a union bound, again, we have
	\begin{align}
	\no \no \mathbb{P}\left(\bigcup\limits_{l=2}^{\frac{n}{s}} \left\{ D\left(\textbf{P}^{(l)},\widetilde{\textbf{P}_{\Pi}^{(l)}}\right)\geq \Delta_n \right\}\right)& \leq \sum\limits_{l=2}^{\frac{n}{s}} \mathbb{P}\left(D\left(\textbf{P}^{(l)},\widetilde{\textbf{P}_{\Pi}^{(l)}}\right)\geq \Delta_n\right)\\
	\no & \leq \frac{n}{s}2s(r-1)e^{-\frac{c}{3}n^{\frac{\alpha}{2}}}\\
	\no &= 2n(r-1)e^{-\frac{c}{3}n^{\frac{\alpha}{2}}}\to 0,\ \
	\end{align}
	as $n \to \infty.$ Thus, for all $l \in \{2, 3, \cdots, \frac{n}{s}\}$, $\widetilde{\textbf{P}_{\Pi}^{(l)}}$'s are close to $\textbf{P}^{(l)}$'s, thus, they will be outside of $\mathcal{F}^{(n)}$ with high probability. Now, we can conclude as $n \to \infty$,
	\begin{align}
	\no \mathbb{P}\left(\bigcup\limits_{l=2}^{\frac{n}{s}} \left\{ D\left(\textbf{P}^{(1)},\widetilde{\textbf{P}_{\Pi}^{(l)}}\right)\leq \Delta_n \right\}\right)\to 0. \ \
	\end{align}
	
	This means that with high probability all $\widetilde{\textbf{P}_{\Pi}^{(l)}}$'s are outside of $\mathcal{F}^{(n)}$, so the adversary can successfully identify Group $1$.
	
	\hspace{-0.18 in}\textbf{Third step: Identifying user $1$ among all of the members of Group $1$:} In this step, we prove that, after identifying Group $1$, the adversary can correctly identify each member. This step can be done using a similar approach to the one above. We define two sets $\mathcal{B}^{(n)}$ and $\mathcal{C}^{(n)}$ around $\textbf{p}_1$. We will show that with high probability, the true estimated value of $\textbf{p}_1$ (shown as $\widetilde{\textbf{p}_1}$) is inside of $\mathcal{B}^{(n)}$. Also, all $\textbf{p}_u$'s of other members of Group 1 are outside of $\mathcal{C}^{(n)}$, and since their estimated values are close to ${\textbf{p}}_u$'s, the estimated values will be outside of $\mathcal{B}^{(n)}$. Therefore, the adversary can successfully invert the permutation $\Pi$ within Group $1$ and identify all of the members. Below are the details.
	
	From (\ref{pu-tilde}) and (\ref{pu-tilde2}), for all $u\in\{1, 2, \cdots, s\}$ and all $i \in \{1, 2, \cdots, r-1\}$, we have
	\begin{align}
	\no \widetilde{p_{u}(i)}=\frac{\abs*{\left\{k: Y_u(k)=i\right\}}}{m},\ \
	\end{align}
	and as a result,
	\begin{align}
	\no \widetilde{p_{\Pi(u)}(i)}=\frac{\abs*{\left\{k: X_u(k)=i\right\}}}{m}=\frac{M_{u}(i)}{m},\ \
	\end{align}
	where $M_{u}(i)= \abs*{\left\{k: X_u(k)=i\right\}}.$
Let's define sets $\mathcal{B}^{(n)}$ and $\mathcal{C}^{(n)}$ as
	\begin{align}
	\no \mathcal{B}^{(n)}= & \bigg\{(x_1, x_2, \cdots, x_{r-1}) \in \mathcal{R}_{\textbf{P}}: | x_i-p_1(i)| \leq\Delta_n, i=1, 2, \cdots, r-1\bigg\},
	\end{align}
	\begin{align}
	\no \mathcal{C}^{(n)}= &\bigg \{(x_1, x_2, \cdots, x_{r-1}) \in \mathcal{R}_{\textbf{P}}:
	|x_i-p_1(i)| \leq 2 \Delta_n, i=1, 2, \cdots, r-1\bigg\},
	\end{align}
	where $\Delta_n = {n^{-\frac{1}{s(r-1)}-\frac{\alpha}{4}}}.$ Figure~\ref{R_P} shows $\textbf{p}_1$, sets $\mathcal{B}^{(n)}$ and $\mathcal{C}^{(n)}$ in range of $\mathcal{R}_{\textbf{P}}$ for case $r=3.$
	\begin{figure}
		\centering
		\includegraphics[width=.5\linewidth, height=0.5 \linewidth]{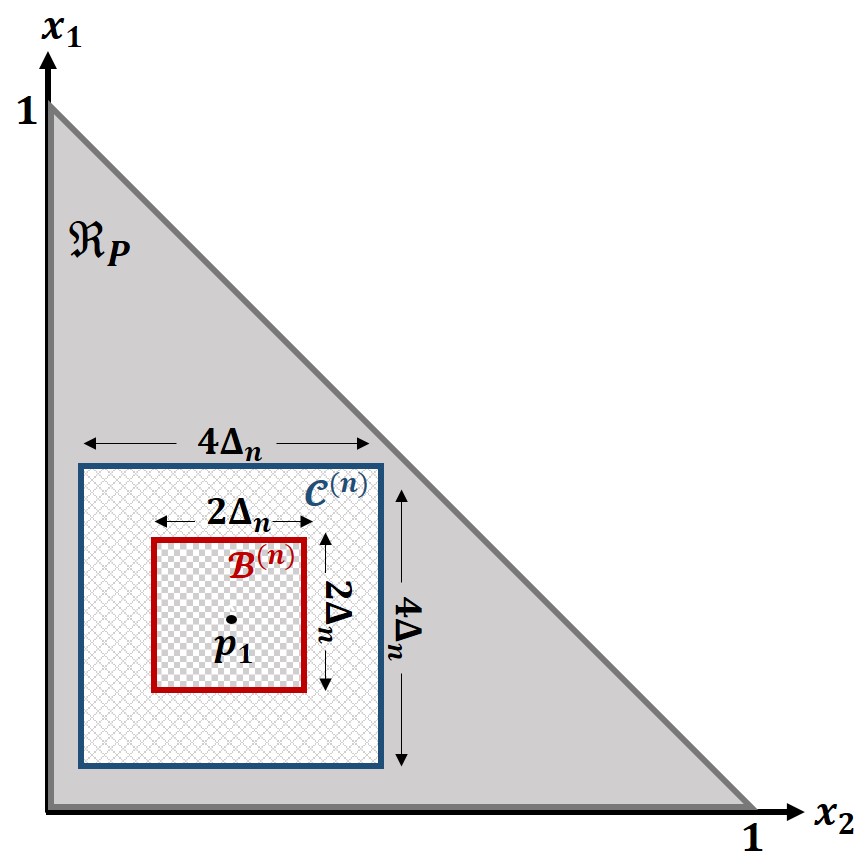}
		\centering
		\caption{$\textbf{p}_1$, sets $\mathcal{B}^{(n)}$ and $\mathcal{C}^{(n)}$ in $\mathcal{R}_\textbf{P}$ for case $r=3$.}
		\label{R_P}
	\end{figure}
Now, we claim for $m =cn^{\frac{2}{s(r-1)} + \alpha}$,
	\begin{enumerate}
		\item $\mathbb{P}\left( {\widetilde{\textbf{p}_{\Pi(1) }}}\in \mathcal{B}^{(n)}\right) \to 1,$
		\item $\mathbb{P}\left( \bigcup\limits_{u=2}^{s} \left\{{\widetilde{\textbf{p}_{\Pi(u)}}}\in \mathcal{B}^{(n)}\right\}\right) \to 0,$
	\end{enumerate}
	as $n \to \infty$.
	Thus, the adversary can identify $\Pi(1)$ by examining $\widetilde{\textbf{p}_u}$'s and choosing the only one that belongs to $\mathcal{B}^{(n)}$.
	
	First, we want to show that as $n$ goes to infinity,
	\[\mathbb{P}\left( {\widetilde{\textbf{p}_{\Pi(1)} }}\in \mathcal{B}^{(n)}\right) \to 1.\]
	For all $i \in \{1, 2, \cdots, r-1\}$, By using (\ref{chern}) and the union bound, we have
	\begin{align}
	\no \mathbb{P}\left({\widetilde{\textbf{p}_{\Pi(1)}}} \notin \mathcal{B}^{(n)}\right) &\leq \sum\limits_{i=1}^{r-1}\mathbb{P}\left(\abs*{\frac{M_1(i)}{m}-p_1(i)} \geq \Delta_n \right)\\
	\no &\leq (r-1)\left(2e^{-\frac{c}{3}n^{\frac{\alpha}{2}}}\right),\ \
	\end{align}	
thus,
	\begin{align}
	\no \mathbb{P}\left({\widetilde{\textbf{p}_{\Pi(1)}}} \in \mathcal{B}^{(n)}\right) &\geq 1-2(r-1)e^{-\frac{c}{3}n^{\frac{\alpha}{2}}}\to 1,\ \
	\end{align}
as $n \to \infty$.

	Now, we need to show that as $n$ goes to infinity,
	\[\mathbb{P}\left( \bigcup\limits_{u=2}^{s} \left\{{\widetilde{\textbf{p}_{\Pi(u)}}}\in \mathcal{B}^{(n)}\right\}\right) \to 0.\]
	First, we show as $n$ goes to infinity,
	\[\mathbb{P}\left(\bigcup\limits_{u=2}^s \left\{\textbf{p}_u \in \mathcal{C}^{(n)}\right\} \right)\to 0. \]
	Note
	\[4 \left(\Delta_n\right)^{r-1} \delta_1 < \mathbb{P}\left( \textbf{p}_u\in \mathcal{C}^{(n)}\right) < 4 \left(\Delta_n\right)^{r-1} \delta_2,\]
	and according to the union bound, for large enough $n$, we have
	\begin{align}
	\no \mathbb{P}\left( \bigcup\limits_{u=2}^s \left\{\textbf{p}_u \in \mathcal{C}^{(n)}\right\} \right) &\leq \sum\limits_{u=2}^s \mathbb{P}\left( \textbf{p}_u\in \mathcal{C}^{(n)}\right) \\
	\nonumber &\leq 4s \left(\Delta_n\right)^{r-1} \delta_2\\
	\nonumber &\leq 4s \frac{1}{n^{\frac{1}{s}+{\frac{\alpha(r-1)}{4}}}} \delta_2 \to 0; \ \
	\end{align}
	thus, all $\textbf{p}_u$'s are outside of $\mathcal{C}^{(n)}$ with high probability.
	
	Now, we claim that given all $\textbf{p}_u$'s are outside of $\mathcal{C}^{(n)}$, $\mathbb{P}\left({\widetilde{\textbf{p}_{\Pi (u)}}} \in \mathcal{B}^{(n)}\right)$ is small.
	Note, for all $i \in \{1, 2, \cdots, r-1\}$, by using (\ref{chern}) and the union bounds, we have
	\begin{align}
	\no \mathbb{P}\left({\widetilde{\textbf{p}_{\Pi(u)}}}\in \mathcal{B}^{(n)}\right) &\leq \mathbb{P}\left(\abs*{{\widetilde{\textbf{p}_{\Pi(u)}}}-\textbf{p}_u } \geq\Delta_n\right) \\
	\no & \leq \sum\limits_{i=1}^{r-1}\mathbb{P}\left(\abs*{{\frac{M_u(i)}{m}}-{p}_u(i) } \geq\Delta_n\right)\\
	\no & \leq 2(r-1)e^{-\frac{c}{3}n^{\frac{\alpha}{2}}}.\ \
	\end{align}
	As a result, by using another union bound, as $n$ becomes large,
	\begin{align}
	\no \mathbb{P}\left( \bigcup\limits_{u=2}^s \left\{\abs*{\widetilde{\textbf{p}_{\Pi(u)}}-\textbf{p}_u}\geq \Delta_n\right\}\right)\leq s\left(2(r-1)e^{-\frac{c}{3}n^{\frac{\alpha}{2}}}\right) \to 0.
	\end{align}
	Thus, for all $u \in\{2, 3, \cdots, s\}$, $\widetilde{\textbf{p}_{\Pi(u)}}$'s are close to $\textbf{p}_u$'s, thus they will be outside of $\mathcal{B}^{(n)}$. Now, we can conclude as $n \to \infty$ that:
	\begin{align}
	\no \mathbb{P}\left( \bigcup\limits_{u=2}^s \left\{{\widetilde{\textbf{p}_{\Pi(u)}}}\in \mathcal{B}^{(n)}\right\}\right) \to 0.
	\end{align}

	Thus, we have proved that if $m=cn^{\frac{2}{s(r-1)}+\alpha}$, there exists an algorithm for the adversary to successfully recover user $1$. Remember, the adversary identifies the members of Group $1$ independent of the structure of the subgraph.
\end{proof}}

\subsection{{\color{black}$r$-State} Markov Chain Model}
\label{markov}
In Sections~\ref{iidr}, we assumed each user's data patterns was i.i.d.; however, in this section, users' data patterns are modeled using Markov chains in which each user's data points are dependent over time. In this model, we again assume there are $r$ possibilities for each users' data point, i.e., $X_u(k) \in \{0,1, \cdots, r-1\}$. More specifically, each user's data set is modeled by a Markov chain with $r$ states. It is assumed that the Markov chains of all users have the same structure but have different transition probabilities. Let $E$ be the set of edges in the assumed transition graph, so, $(i, j) \in E$ if there exists an edge from state $i$ to state $j$, meaning that $p_u(i,j)=\mathbb{P}\left(X_u(k+1)=j|X_u(k)=i\right)>0$. The transition matrix is a square matrix used to describe the transitions of a Markov chain; thus, different users can have different transition probability matrices. Note for each state $i$, we have
$\sum\limits_{j=1}^{r-1} p_u(i,j)=1,$
so, the adversary can focus on a subset of size $d=|E|-r$ of the transition probabilities for recovering the entire transition matrix. Let $\textbf{p}_u$ be the vector that contains these transition probabilities for user $u$. We write
	\[\textbf{p}_u= \begin{bmatrix}
p_u(1) \\ p_u(2) \\ \vdots \\p_u(|E|-r) \end{bmatrix}, \ \ \ \textbf{p} =\left[\textbf{p}_{1}\ \ \textbf{p}_{2}\ \ \cdots\ \ \textbf{p}_{n}\right].
\]

We also consider all $p_u(i)$'s are drawn independently from some continuous density function, $f_\textbf{P}(\textbf{p}_u)$, on the $(0,1)^{|E|-r}$ hypercube. Define the range of distribution as
\begin{align}
	\no \mathcal{R}_{\textbf{P}} &= \left\{ (x_1, x_2, \cdots, x_{|E|-r}) \in (0,1)^{|E|-r}: x_i > 0 , x_1+x_2+\cdots+x_{|E|-r} < 1\right\},
\end{align}
and as before, we assume there are $ \delta_1, \delta_2 >0$, such that
\begin{equation}
\begin{cases}
\no \delta_1\leq f_{\textbf{P}}(\textbf{p}_u) \leq \delta_2, & \textbf{p}_u \in \mathcal{R}_{\textbf{p}}.\\
f_{\textbf{P}}(\textbf{p}_u)=0, & \textbf{p}_u \notin \mathcal{R}_{\textbf{p}}.
\end{cases}
\end{equation}
Now, we can repeat the similar steps as the previous sections to prove the following theorem.
\begin{thm}\label{markov_thm}
	For an irreducible, aperiodic Markov chain model, if ${\textbf{Y}}$ is the anonymized version of $\textbf{X}$ as defined above, the size of the group including user $1$ is $s$, and
	\begin{itemize}
		\item {\color{black}$m=\Omega\left(n^{\frac{2}{s(|E|-r)}+\alpha}\right)$}, for any $\alpha > 0$;
	\end{itemize}
then, user $1$ has no privacy at time $k$.
\end{thm}

\begin{proof}
The basic ideas behind the proof of Theorem~\ref{markov_thm} are similar to the ones for Theorems~\ref{r_state_thm}; thus, in this part we just focus on the differences and key ideas.

Define the random variable $M_u(i)$ as the total number of visits by user $u$ to state $i$, for all $u\in \{ 1, 2, \cdots, n\}$ and $i\in \{0, 1, \cdots, r-1\}$. Since the Markov chain is irreducible and aperiodic, and $m \to \infty$, all $\frac{M_i(u)}{m}$ converge to their stationary values{\color{black}~\cite{principles_Levy}}. Given $M_u(i) = m_u(i)$, the transitions from state $i$ to state $j$ for user $u$ has a multinomial distribution with probabilities $p_u(i,j)$. Now, considering the fact that the vector $\textbf{p}_u$ uniquely determines the user $u$, the adversary can invert the anonymization permutation function in a similar way to the i.i.d.\ case by focusing on $\textbf{p}_u$'s.
Let
\[\mathbf{\Phi}=\left[\mathbf{\Phi}_1\ \ \mathbf{\Phi}_2 \ \ \cdots \ \ \mathbf{\Phi}_s\right],\]
\[\mathbf{\Psi}=\left[\mathbf{\Psi}_1\ \ \mathbf{\Psi}_2\ \ \cdots \ \ \mathbf{\Psi}_s\right],\]
where $\mathbf{\Phi}_u \in \mathbb{R}^{|E|-r}$ and $\mathbf{\Psi}_u \in \mathbb{R}^{|E|-r}$. Define
\begin{align}
\no D\left(\mathbf{\Phi},\mathbf{\Psi}\right)=\min\limits_{\sigma \in \Sigma_s}\left \{\max\big \{||\mathbf{\Phi}_1-\mathbf{\Psi}_{\sigma(1)}||_{\infty}, ||\mathbf{\Phi}_2-\mathbf{\Psi}_{\sigma(2)}||_{\infty}, \cdots, ||\mathbf{\Phi}_s-\mathbf{\Psi}_{\sigma(s)}||_{\infty}\big \}\right\},
\end{align}
where for $u \in \{1, 2, \cdots, s\}$,
{\begin{align}
	\nonumber ||\mathbf{\Phi}_u-\mathbf{\Psi}_{\sigma(u)}||_{\infty}=\max\left\{|\Phi_u(i)-\Psi_{\sigma(u)}(i)|:i = 1, 2, \cdots, |E|-r\right\}. \
	\end{align}
}
and we claim for $m=cn^{\frac{2}{s(|E|-r)}+\alpha}$ and large enough $n$,
\begin{itemize}
	\item $\mathbb{P}\left(D\left(\textbf{P}^{(1)}, \widetilde{\textbf{P}_{\Pi}^{(1)}}\right)\leq \Delta'_n\right) \to 1,$
	\item $\mathbb{P}\left(\bigcup\limits_{l=2}^{\frac{n}{s}} \left\{ D\left(\textbf{P}^{(1)},\widetilde{\textbf{P}_{\Pi}^{(l)}}\right)\leq \Delta'_n \right\}\right) \to 0$ ,
\end{itemize}
where $\Delta'_n={n^{-\frac{1}{s(|E|-r)}-\frac{\alpha}{4}}}$. This can be shown similar to the proof of Theorem~\ref{r_state_thm}. First, define $\mathcal{F}'^{(n)}$ and $\mathcal{H}'^{(n)}$ as
\begin{align}
\no \mathcal{F}'^{(n)}= & \left\{(\textbf{x}_1, \textbf{x}_2, \cdots ,\textbf{x}_{s}) \in \left(\mathbb{R}^{(|E|-r)}\right)^s:
\max\limits_u\{|\textbf{x}_u-\textbf{p}_u|\}\leq \Delta'_n, u=1, 2, \cdots ,s\right\};
\end{align}
\begin{align}
\no \mathcal{H}'^{(n)}= &\left\{(\textbf{x}_1, \textbf{x}_2, \cdots ,\textbf{x}_{s}) \in \left(\mathbb{R}^{(|E|-r)}\right)^s:\max\limits_u\{|\textbf{x}_u-\textbf{p}_u|\} \leq 2\Delta'_n, u=1, 2, \cdots ,s\right\};
\end{align}
then, prove that the adversary can identify Group $1$ successfully.

In the next step, the adversary has to identify each member of Group $1$ correctly.
Define sets $\mathcal{B}'^{(n)}$ and $\mathcal{C}'^{(n)}$ as
\begin{align}
	\no \mathcal{B}'^{(n)}= & \bigg\{(x_1, x_2, \cdots ,x_{d}) \in \mathcal{R}_{\textbf{P}}: | x_i-p_1(i)| \leq \Delta'_n, i=1,2, \cdots , d\bigg\},
\end{align}
\begin{align}
	\no \mathcal{C}'^{(n)}= &\bigg\{(x_1, x_2, \cdots, x_{d}) \in \mathcal{R}_{\textbf{P}}:
	|x_i-p_1(i)| \leq 2 \Delta'_n, i=1, 2, \cdots, d\bigg\},
\end{align}
where $\Delta'_n = {n^{-\frac{1}{s(|E|-r)}-\frac{\alpha}{4}}}.$ Now, we claim for $m =cn^{\frac{2}{s(|E|-r)} + \alpha}$,
\begin{enumerate}
	\item $\mathbb{P}\left( {\widetilde{\textbf{p}_{\Pi(1) }}}\in \mathcal{B}'^{(n)}\right) \to 1,$
	\item $\mathbb{P}\left( \bigcup\limits_{u=2}^{s} \left\{{\widetilde{\textbf{p}_{\Pi(u)}}}\in \mathcal{B}'^{(n)}\right\}\right) \to 0,$
\end{enumerate}
as $n \to \infty $. This can be shown similar to the proof of Theorem~\ref{r_state_thm}, so the adversary can successfully recover data traces of user $1$.
\end{proof}
{\color{black}
	\hspace{-0.18 in} \textbf{{Discussion $5$}:} Note that the i.i.d.\ case can also be written as a Markov chain with a transition matrix with identical rows; then, $|E|=r^2$. However, for the i.i.d.\ case, if the adversary knows $r-1$ elements of a row, they know that row and all of the others. In other words, if we restrict the users' data models to i.i.d., then we are using a different model where $p_u(i)$'s are restricted in a way to create an i.i.d.\ sequence. This is a different model and is not compatible to our model for the Markov chain where $p_u(i)$'s are drawn independently from some continuous density function, $f_\textbf{P}(\textbf{p}_u)$, on the $(0,1)^{|E|-r}$ hypercube. Thus, the results of Theorem 2 cannot be applied to the i.i.d.\ case.
}

\section{Impact of Dependency on Privacy using Anonymization and Obfuscation}
\label{obfs}
Here, we consider the case when both anonymization and obfuscation techniques are employed, as shown in Figure~\ref{fig:xyz}. We assume similar obfuscation to~\cite{Nazanin_IT}. To obfuscate the users' data points, for each user $u$, we independently generate a random variable $R_u$ that is uniformly distributed between $0$ and $a_n$, where $a_n \in (0,1]$. The value of $R_u$ shows the probability that the user's data point is changed to a different value by obfuscation, and $a_n$ is termed the ``noise level' of the system.
{\color{black}Let $\textbf{Z}_u$ be the vector that contains
the obfuscated version of user $u$'s data points, and ${\textbf{Z}}$ be the collection of ${\textbf{Z}}_u$ for all users,
\[{\textbf{Z}}_u = \begin{bmatrix}
{Z}_u(1) \\ {Z}_u(2) \\ \vdots \\ {Z}_u(m) \end{bmatrix} , \ \ \ {\textbf{Z}} =\left[ {\textbf{Z}}_{1} \ \ {\textbf{Z}}_{2}\ \ \cdots \ \ {\textbf{Z}}_{n}\right].
\]
Thus, the adversary's observation ${\textbf{Y}}$ is the anonymized version of $\textbf{Z}$;
\begin{align}
\no {\textbf{Y}} &=\textrm{Perm}\left(\textbf{Z}_{1}, \textbf{Z}_{2}, \cdots, \textbf{Z}_{n}; \Pi \right) \\
\nonumber &=\left[ \textbf{Z}_{\Pi^{-1}(1)}\ \ \textbf{X}_{\Pi^{-1}(2)}\ \ \cdots \ \ \textbf{Z}_{\Pi^{-1}(n)}\right ] \\
\nonumber &=\left[ {\textbf{Y}}_{1}\ \ {\textbf{Y}}_{2}\ \ \cdots\ \ {\textbf{Y}}_{n}\right]. \ \
\end{align}
}

\subsection{{\color{black}$r$-State} i.i.d.\ Model}
\label{iidr_obf}
Now, assume users' data points can have $r$ possibilities $\left(0, 1, \cdots, r-1\right)$. Similar to Section~\ref{iidr}, we assume $\textbf{p}_u$'s are drawn independently from some continuous density function, $f_\textbf{P}(\textbf{p}_u)$, which has support on a subset of the $(0,1)^{r-1}$ hypercube, and $\textbf{p}_u$, $f_\textbf{P}(\textbf{p}_u)$, and $\mathcal{R}_\textbf{P}$ are defined as in Section~\ref{iidr}.

{\color{black}
To create a noisy version of data samples, for each user $u$, we independently generate a
random variable $R_u$ that is uniformly distributed between $0$ and $a_n$, where $a_n \in (0,1]$\footnote{{\color{black}It is desirable that our results are true over the largest set of strategies that users can employ. In fact, our results would apply to a general set of distributions and are true for any random noise with support that extends out to the maximum amount of $a_n$. The reason that we have used a uniformly random noise is that we want to have a similar mechanism as~\cite{Nazanin_IT} to have a good comparison between the results of this paper and~\cite{Nazanin_IT} to show that dependency is a significant detriment to the privacy of users.}}.} Then, the obfuscated data is obtained by passing the users' data through an $r$-ary symmetric channel with a random error probability $R_u$, so for $j \in \{ 0, 1, \cdots, r-1\}$:
\[
\mathbb{P}({Z}_{u}(k)=j| X_{u}(k)=i) =\begin{cases}
1-R_u, & \textrm{for } j=i.\\
\frac{R_u}{r-1}, & \textrm{for } j \neq i.
\end{cases}
\]

The effect of the obfuscation is to alter the probability distribution function of each user's data points in a way that is unknown to the adversary, since it is independent of all past activity of the user, and hence, the obfuscation inhibits user identification. For each user, $R_u$ is generated once and is kept constant for the collection of data points of length $m$, thus providing a very low-weight obfuscation algorithm. {\color{black}

{\color{black}
Now, define
$$Q_u(i)=\mathbb{P}\left(Z_u(k)=i\right),$$
where
\begin{align}
\no Q_u(i)&=P_u(i)(1-R_u)+(1-P_u(i))R_u\\
&=P_u(i)+(1-2P_u(i))R_u.\ \
\label{Q}
\end{align}	
The vectors $\textbf{Q}_u$ and $\textbf{Q}$ which contain the obfuscated probabilities are defined as below:
\[\textbf{Q}_u= \begin{bmatrix}
Q_u(1) \\ Q_u(2) \\ \vdots \\Q_u(r-1) \end{bmatrix}, \ \ \ \textbf{Q} =\left[\textbf{Q}_{1}\ \ \textbf{Q}_{2}\ \ \cdots\ \ \textbf{Q}_{n}\right],
\]
and the vector containing the permutation of those probabilities after anonymization is $\textbf{W}$. Thus,
\begin{align}
\no {\textbf{W}} &=\textrm{Perm}\left(\textbf{Q}_1, \textbf{Q}_2, \cdots, \textbf{Q}_n; \Pi \right) \\
\nonumber &=\left[\textbf{Q}_{\Pi^{-1}(1)}\ \ \textbf{Q}_{\Pi^{-1}(2)}\ \ \cdots\ \ \textbf{Q}_{\Pi^{-1}(n)}\right ] \\
\nonumber &=\left[{\textbf{W}}_{1}\ \ {\textbf{W}}_{2}\ \ \cdots\ \ {\textbf{W}}_{n}\right]. \ \
\end{align}
}

\begin{thm}\label{r_state_thm_converse_obfs}
	For the above $r$-state model, if ${\textbf{Z}}$ is the obfuscated version of $\textbf{X}$, and ${\textbf{Y}}$ is the anonymized version of ${\textbf{Z}}$ as defined above, the size of the group including user $1$ is $s$, and
	\begin{itemize}
		\item {\color{black}$m =\Omega\left(cn^{\frac{2}{s(r-1)} + \alpha}\right)$} for any $\alpha>0$;
		\item $R_u \sim \Unif [0, a_n]$, where {\color{black}$a_n = O\left(n^{-\frac{1}{s(r-1)}-\beta}\right)$} for any $\beta>\frac{\alpha}{4}$;
	\end{itemize}
then, user $1$ has no privacy at time $k$.
\end{thm}

		\hspace{-0.18 in}\textbf{{Discussion $6$}:}
It is insightful to compare this result to~\cite[Theorem 2]{Nazanin_IT}. We can see that when users' traces are dependent, the required level of obfuscation and anonymization to achieve privacy is significantly higher.
Therefore, we see that dependency can significantly reduce the privacy of users. However, note that the asymptotic noise level is still zero in this case. Specifically, if $A_m(u) = \frac{\abs*{\left\{k:{Z}_u(k) \neq X_u(k)\right\}}}{m}$, then
\[\mathbb{E}[A_m(u)]=\mathbb{E}[A_m]= O\left({n^{-\frac{1}{s(r-1)}-\beta}}\right) \to 0,\]
implying that the asymptotic noise level is zero.
$\\$

{\color{black}
\noindent \textbf{Proof of Theorem~\ref{r_state_thm_converse_obfs}:}
\begin{proof}
	The proof of Theorem~\ref{r_state_thm_converse_obfs} is similar to the proof of Theorem~\ref{r_state_thm} and consists of three parts:
	\begin{itemize}
		\item \textbf{First step:} Showing the adversary can reconstruct the association graph of the obfuscated and anonymized version of data with an arbitrarily small error probability.
		\item \textbf{Second step:} Showing the adversary can uniquely identify Group $1$ with an arbitrarily small error probability.
		\item \textbf{Third step:} Showing the adversary can successfully identify all of the members of Group $1$ with an arbitrarily small error probability.
	\end{itemize}

	\hspace{-0.18 in}\textbf{First step: Reconstruction of the association graph:} In Lemma~\ref{lem1}, we show that for the case of anonymization, the adversary can reconstruct the entire association graph of the anonymized data with an arbitrarily small error probability if the number of the adversary's observations per user $(m)$ is bigger than $(\log n)^3$. Since obfuscation is done independently (from other users' obfuscation and from users' data), it does not change the association graph. Therefore, since $n^{\frac{2}{s(r-1)}}>(\log n)^3$, we can use Lemma~\ref{lem1} to show the adversary can reconstruct the association graph of the obfuscated and anonymized data with an arbitrarily small error probability.

\hspace{-0.18 in}\textbf{Second step: Identifying Group $1$ among all of the groups:} Now, assume the size of Group $1$ is $s$. Without loss of generality, suppose the members of Group $1$ are users $\{1,2, \cdots , s\}$, so there are at most $\frac{n}{s}$ groups of size $s$. We call these Groups $1, 2, \cdots, \frac{n}{s}$. The adversary needs to first identify the Group $1$ among all of these groups.
	
	According to Section~\ref{iidr}, $\Sigma_s$ is defined as the set of all permutation on $s$ elements, and $\textbf{P}^{(l)}$ is a vector which contains probability distributions of users belong to Group $l$, and $\widetilde{\textbf{P}_{\Pi}^{(l)}}$ is a vector which contains the estimate of adversary about the probability distribution of users belong to Group $l$. For example, For Group $1$, we have
	\[\textbf{P}^{(1)}=\left[\textbf{p}_1 \ \ \textbf{p}_2 \ \ \cdots \ \ \textbf{p}_s\right],\]
	and
	\[\widetilde{\textbf{P}_{\Pi}^{(1)}}=\left[\widetilde{\textbf{p}_{\Pi(1)}} \ \ \widetilde{\textbf{p}_{\Pi(2)}} \ \ \cdots \ \ \widetilde{\textbf{p}_{\Pi(s)}}\right].\] 		
	We claim for $m=cn^{\frac{2}{s(r-1)}+\alpha}$, $a_n = c'n^{-\left(\frac{1}{s(r-1)}+\beta\right)}$, and large enough $n$,
	\begin{itemize}
		\item $\mathbb{P}\left(D\left(\textbf{P}^{(1)}, \widetilde{\textbf{P}_{\Pi}^{(1)}}\right)\leq \Delta_n\right) \to 1, $
		\item $\mathbb{P}\left(\bigcup\limits_{l=2}^{\frac{n}{s}} \left\{ D\left(\textbf{P}^{(1)},\widetilde{\textbf{P}_{\Pi}^{(l)}}\right)\leq \Delta_n \right\}\right) \to 0$ ,
	\end{itemize}
	where $\Delta_n= {n^{-\frac{1}{s(r-1)}-\frac{\alpha}{4}}}$. As in Section~\ref{iidr},
	\begin{align}
	\no \mathcal{F}^{(n)}= & \left\{(\textbf{x}_1, \textbf{x}_2, \cdots ,\textbf{x}_{s}) \in \left(\mathbb{R}^{(r-1)}\right)^s:
	\max\limits_u\{|\textbf{x}_u-\textbf{p}_u|\}\leq \Delta_n, u=1, 2, \cdots ,s\right\},
	\end{align}
	\begin{align}
	\no \mathcal{H}^{(n)}= &\left\{(\textbf{x}_1, \textbf{x}_2, \cdots ,\textbf{x}_{s}) \in \left(\mathbb{R}^{(r-1)}\right)^s:\max\limits_u\{|\textbf{x}_u-\textbf{p}_u|\} \leq 2\Delta_n, u=1, 2, \cdots ,s\right\}.
	\end{align}

	First, we prove, for large enough $n$,
	\[\mathbb{P}\left(D\left(\textbf{P}^{(1)},\widetilde{\textbf{P}_{\Pi}^{(1)}}\right)\leq \Delta_n\right) \to 1.\]
	Note that for all $u\in\{1, 2, \cdots, n\}$ and all $i \in \{1, 2, \cdots, r-1\}$, the adversary computes $\widetilde{p_u(i)}$ as follow:
\begin{align}
\widetilde{p_{u}(i)}=\frac{\abs*{\left\{k: Y_u(k)=i\right\}}}{m},\ \
\label{tild3}
\end{align}
and as a result,
\begin{align}
\widetilde{p_{\Pi(u)}(i)}=\frac{\abs*{\left\{k: Z_u(k)=i\right\}}}{m}=\frac{\widebreve{M}_{u}(i)}{m},\ \
\label{tild4}
\end{align}
where $\widebreve{M}_{u}(i)= \abs*{\left\{k: Z_u(k)=i\right\}}.$
	Now, for all $u \in \{1,2, \cdots,n\}$ and all $i \in \{0,1, \cdots,r-1\}$, we have
	\begin{align}
	\no \mathbb{P}\left(\bigg |\frac{\widebreve{M}_u(i)}{m}-p_u(i)\bigg |\leq \Delta_n\right) &= \mathbb{P}\left(p_u(i)-\Delta_n\leq \frac{\widebreve{M}_u(i)}{m}\leq p_u(i)+\Delta_n \right)\\
	\nonumber &= \mathbb{P}\left(p_u(i)-\Delta_n-q_u(i)\leq\frac{\widebreve{M}_u(i)}{m}-q_u(i)\leq p_u(i)+\Delta_n-q_u(i) \right).\ \
	\end{align}
	Note that for all $u \in \{1,2, \cdots,n\}$ and all $i \in \{0,1, \cdots,r-1\}$, we have
	\begin{align}
	\no \abs*{p_u(i)-q_u(i)} &=|1-2p_u(i)|R_u\\
	\no & \leq R_u \leq a_n,
	\end{align}
	so, we can conclude for all $u \in \{1,2, \cdots,n\}$ and all $i \in \{0,1, \cdots,r-1\}$,
	\begin{align}
	\no \mathbb{P}\left(\bigg |\frac{\widebreve{M}_u(i)}{m}-p_u(i)\bigg |\leq \Delta_n\right) &= \mathbb{P}\left(p_u(i)-\Delta_n-q_u(i)\leq \frac{\widebreve{M}_u(i)}{m}-q_u(i)\leq p_u(i)+\Delta_n-q_u(i) \right)\\
	\nonumber &\geq\mathbb{P}\left(-\Delta_n+ a_n\leq \frac{\widebreve{M}_u(i)}{m}-q_u(i)\leq -a_n+\Delta_n \right)\\
	&= \mathbb{P}\left(\abs*{\frac{\widebreve{M}_u(i)}{m}-q_u(i)}\leq m (\Delta_n-a_n) \right).\ \
	\label{obfs1}
	\end{align}
	By employing a Chernoff bound, we have
	\begin{align}
	\no \mathbb{P}\left(\abs*{\frac{\widebreve{M}_u(i)}{m}-q_u(i)}\leq m (\Delta_n-a_n) \right) &\geq 1-2e^{-\frac{(\Delta_n-a_n)^2}{3q_u(i)}} \\
	\no &\geq 1-2e^{-\left(\frac{1}{3q_u(i)}\right)\left(cn^{\frac{2}{s(r-1)}+\alpha}\right)\left(\frac{1}{n^{\frac{1}{s(r-1)}+\frac{\alpha}{4}}}- \frac{c'}{n^{\frac{1}{s(r-1)} + \beta}}\right)^2 }\\
	&\geq 1-2e^{-\frac{1}{3}\left(cn^{\frac{2}{s(r-1)}+\alpha}\right)\left(\frac{1}{n^{\frac{1}{s(r-1)}+\frac{\alpha}{4}}}- \frac{c'}{n^{\frac{1}{s(r-1)} + \beta}}\right)^2}\ \
	\label{obfs2}
	\end{align}
	Now from (\ref{obfs1}) and (\ref{obfs2}), we can conclude for all $u \in \{1,2, \cdots,n\}$ and all $i \in \{0,1, \cdots,r-1\}$,
	\begin{align}
	\no \mathbb{P}\left(\bigg |\frac{\widebreve{M}_u(i)}{m}-p_u(i)\bigg |\leq \Delta_n\right) \geq 	 1-2e^{-\frac{1}{3}\left(cn^{\frac{2}{s(r-1)}+\alpha}\right)\left(\frac{1}{n^{\frac{1}{s(r-1)}+\frac{\alpha}{4}}}- \frac{c'}{n^{\frac{1}{s(r-1)} + \beta}}\right)^2},\ \
	\end{align}		
	and
	\begin{align}
 \mathbb{P}\left(\bigg |\frac{\widebreve{M}_u(i)}{m}-p_u(i)\bigg |\geq \Delta_n\right) &\leq 	 2e^{-\frac{1}{3}\left(cn^{\frac{2}{s(r-1)}+\alpha}\right)\left(\frac{1}{n^{\frac{1}{s(r-1)}+\frac{\alpha}{4}}}- \frac{c'}{n^{\frac{1}{s(r-1)} + \beta}}\right)^2}.\ \
	\label{obfs3}
	\end{align}	
	
	Now, for all $u \in \text{Group }1$, all $i \in \{1, 2, \cdots, r-1\}$ and any $\beta>\frac{\alpha}{4}$, (\ref{obfs3}) and the union bound yield
	\begin{align}
	\no \mathbb{P}\left(D\left(\textbf{P}^{(1)},\widetilde{\textbf{P}_{\Pi}^{(1)}}\right)\geq \Delta_n\right) &\leq \sum\limits_{u=1}^{s}\sum\limits_{i=1}^{r-1} \mathbb{P}\left(\bigg |\frac{\widebreve{M}_u(i)}{m}-p_u(i)\bigg |\geq \Delta_n\right) \\
	\no &\leq 2s(r-1)e^{-\frac{1}{3}\left(cn^{\frac{2}{s(r-1)}+\alpha}\right)\left(\frac{1}{n^{\frac{1}{s(r-1)}+\frac{\alpha}{4}}}- \frac{c'}{n^{\frac{1}{s(r-1)} + \beta}}\right)^2} \to 0,\ \
	\end{align}
	as $n \to \infty.$
	As a result,
	\begin{align}
	\no \mathbb{P}\left(D\left(\textbf{P}^{(1)},\widetilde{\textbf{P}_{\Pi}^{(1)}}\right)\leq \Delta_n\right) \to 1,\ \
	\end{align}
as $n \to \infty$.

	In the next step, we prove $\mathbb{P}\left(\bigcup\limits_{l=2}^{\frac{n}{s}} \left\{ D\left(\textbf{P}^{(1)},\widetilde{\textbf{P}_{\Pi}^{(l)}}\right)\leq \Delta_n \right\}\right) \to 0$. For all groups other than Group $1$, we have
	\[(4 \Delta_n)^{s(r-1)}\delta_1\leq \mathbb{P}\left( \textbf{P}^{(l)} \in \mathcal{H}'^{(n)}\right) \leq (4 \Delta_n)^{s(r-1)}\delta_2,\]
	and as a result,
	\begin{align}
	\no \mathbb{P}\left(\textbf{P}^{(l)} \in \mathcal{H}^{(n)}\right) &\leq \delta_2 (4\Delta_n)^{s(r-1)}\\
	\no &=\delta_2 4^{s(r-1)}\frac{1}{n^{1+\frac{\alpha}{4}s(r-1)}}.\ \
	\end{align}
	Similarly, for any $\sigma \in \Sigma_s$,
	\begin{align}
	\no \mathbb{P}\left(\textbf{P}^{(l)}_\sigma \in \mathcal{H}'^{(n)}\right) &\leq \delta_2 (4\Delta_n)^{s(r-1)}\\
	\no &=\delta_2 4^{s(r-1)}\frac{1}{n^{1+\frac{\alpha}{4}s(r-1)}},\ \
	\end{align}
	and since $|\Sigma_s|=s!$, by a union bound,
	\begin{align}
	\no \mathbb{P}\left(\bigcup\limits_{l=2}^{\frac{n}{s}} \left\{ \bigcup\limits_{\sigma \in \Sigma_s} \left\{ \textbf{P}^{(l)}_\sigma \in \mathcal{H}^{(n)}\right\} \right\}\right) &\leq \sum\limits_{l=2}^{\frac{n}{s}} \sum\limits_{\sigma \in \Sigma_s} \mathbb{P}\left(\textbf{P}^{(l)}_\sigma \in \mathcal{H}^{(n)} \right)\\
	\no &\leq \frac{n}{s}s!\delta_2 4^{s(r-1)}\frac{1}{n^{1+\frac{\alpha}{4}s(r-1)}}\\
	\no &=(s-1)!4^{s(r-1)}\delta_2{n^{-\frac{\alpha}{4}s(r-1)}}\to 0,\ \
	\end{align}
	as $n \to \infty.$
	Thus, all $\textbf{P}^{(l)}$'s are outside of $\mathcal{H}^{(n)}$ with high probability.
	
	Now, we claim that given all $\textbf{P}^{(l)} $'s are outside of $\mathcal{H}^{(n)}$, $\mathbb{P}\left(\bigcup\limits_{l=2}^{\frac{n}{s}} \left\{ D\left(\textbf{P}^{(1)},\widetilde{\textbf{P}_{\Pi}^{(l)}}\right)\leq \Delta_n \right\}\right)$ is arbitrarily small. In other words, by using a Chernoff bound, it is shown
	$\widetilde{\mathbf{P}^{(l)}}$'s are close to $\mathbf{P}^{(l)}$'s, and they will be outside of $\mathcal{F}^{(n)}$. Thus, for all $u \in \text{Group }l$ and all $i \in \{1, 2, \cdots, r-1\}$, (\ref{obfs3}) and the union bound yield
	\begin{align}
	\no \mathbb{P}\left(D\left(\textbf{P}^{(1)},\widetilde{\textbf{P}_{\Pi}^{(l)}}\right)\leq \Delta_n\right)&=\mathbb{P}\left(D\left(\textbf{P}^{(l)},\widetilde{\textbf{P}_{\Pi}^{(l)}}\right)\geq \Delta_n\right) \\
	\no &\leq \sum\limits_{u=1}^{s}\sum\limits_{i=1}^{r-1} \mathbb{P}\left(\bigg{|}\frac{\widebreve{M}_u(i)}{m}-p_u(i))\bigg{|} \leq \Delta_n \right) \\\no &\leq 2s(r-1)e^{-\frac{1}{3}\left(cn^{\frac{2}{s(r-1)}+\alpha}\right)\left(\frac{1}{n^{\frac{1}{s(r-1)}+\frac{\alpha}{4}}}- \frac{c'}{n^{\frac{1}{s(r-1)} + \beta}}\right)^2}.\ \
	\end{align}
	Now, by using a union bound again, we can conclude, for any $\beta>\frac{\alpha}{4}$,
	\begin{align}
	\no \no \mathbb{P}\left(\bigcup\limits_{l=2}^{\frac{n}{s}} \left\{ D\left(\textbf{P}^{(l)},\widetilde{\textbf{P}_{\Pi}^{(l)}}\right)\leq \Delta_n \right\}\right)& \leq \sum\limits_{l=2}^{\frac{n}{s}} \mathbb{P}\left(D\left(\textbf{P}^{(l)},\widetilde{\textbf{P}_{\Pi}^{(l)}}\right)\leq \Delta_n\right)\\
	\no & \leq 2n(r-1)e^{-\frac{1}{3}\left(cn^{\frac{2}{s(r-1)}+\alpha}\right)\left(\frac{1}{n^{\frac{1}{s(r-1)}+\frac{\alpha}{4}}}- \frac{c'}{n^{\frac{1}{s(r-1)} + \beta}}\right)^2 }\to 0,\ \
	\end{align}
	as $n \to \infty$. Thus, we have shown that for all $l \in \{1, 2, \cdots,\frac{n}{s}\}$, $\widetilde{\mathbf{P}^{(l)}}$'s are close to ${\mathbf{P}^{(l)}}$, which are outside of set $\mathcal{F}^{(n)}$. As a result, as $n \to \infty$,
	\begin{align}
	\no \no \mathbb{P}\left(\bigcup\limits_{l=2}^{\frac{n}{s}} \left\{ D\left(\textbf{P}^{(1)},\widetilde{\textbf{P}_{\Pi}^{(l)}}\right)\leq \Delta_n \right\}\right) \to 0.
	\end{align}
	
	\hspace{-0.18 in}\textbf{Third step: Identifying user $1$ among all of the members of Group $1$:} In this step, we need to prove that after identifying Group $1$, the adversary can correctly identify each member. In other words, the adversary should identify the permutation of Group $1$.
	
	From (\ref{tild3}) and (\ref{tild4}), for all $u\in \{1, 2, \cdots, s\}$ and all $i \in\{1,2, \cdots, r-1\}$, we have
	\begin{align}
	\no \widetilde{p_{u}(i)}=\frac{\abs*{\left\{k: Y_u(k)=i\right\}}}{m},\ \
	\end{align}
	and as a result,
	\begin{align}
	\no \widetilde{p_{\Pi(u)}(i)}=\frac{\abs*{\left\{k: Z_u(k)=i\right\}}}{m}=\frac{\widebreve{M}_{u}(i)}{m},\ \
	\end{align}
	where $\widebreve{M}_{u}(i)= \abs*{\left\{k: Z_u(k)=i\right\}}.$

	As in Section~\ref{iidr}, we define sets $\mathcal{B}^{(n)}$ and $\mathcal{C}^{(n)}$ as
	\begin{align}
	\no \mathcal{B}^{(n)}= & \bigg\{(x_1, x_2, \cdots, x_{r-1}) \in \mathcal{R}_{\textbf{P}}: | x_i-p_1(i)| \leq\Delta'_n, i=1, 2, \cdots, r-1\bigg\},
	\end{align}
	\begin{align}
	\no \mathcal{C}^{(n)}= &\bigg \{(x_1, x_2, \cdots, x_{r-1}) \in \mathcal{R}_{\textbf{P}}:
	|x_i-p_1(i)| \leq 2 \Delta'_n, i=1, 2, \cdots, r-1\bigg\},
	\end{align}
	where $\Delta_n = {n^{-\frac{1}{s(r-1)}-\frac{\alpha}{4}}}.$
	We claim that for $m =cn^{\frac{2}{s(r-1)} + \alpha}$ and $a_n = c'n^{-\left(\frac{1}{s(r-1)}+\beta\right)}$,
	\begin{enumerate}
		\item $\mathbb{P}\left( {\widetilde{\textbf{p}_{\Pi(1) }}}\in \mathcal{B}^{(n)}\right) \to 1,$
		\item $\mathbb{P}\left( \bigcup\limits_{u=2}^{s} \left\{{\widetilde{\textbf{p}_{\Pi(u)}}}\in \mathcal{B}^{(n)}\right\}\right) \to 0,$
	\end{enumerate}
	as $n \to \infty$.
	Thus, the adversary can identify $\Pi(1)$ by examining $\widetilde{\textbf{p}_u}$'s and choosing the only one that belongs to $\mathcal{B}^{(n)}$.

	First, we show that as $n$ goes to infinity,
	\[\mathbb{P}\left( {\widetilde{{\textbf{p}}_{\Pi(1) }}}\in \mathcal{B}^{(n)}\right) \to 1.\]
	According to (\ref{obfs3}) and the union bound, for all $u \in $ Group 1 and all $i \in \{1,2, \cdots, r-1\}$, we have
	\begin{align}
	\no \mathbb{P}\left({\widetilde{\textbf{p}_{\Pi(1)}}} \notin \mathcal{B}^{(n)}\right) &\leq \sum\limits_{i=1}^{r-1}\mathbb{P}\left(\abs*{\frac{\widebreve{M}_1(i)}{m}-p_1(i)} \geq \Delta_n \right)\\
	\no &\leq (r-1)\left(2e^{-\frac{1}{3}\left(cn^{\frac{2}{s(r-1)}+\alpha}\right)\left(\frac{1}{n^{\frac{1}{s(r-1)}+\frac{\alpha}{4}}}- \frac{c'}{n^{\frac{1}{s(r-1)} + \beta}}\right)^2}\right),\ \
	\end{align}
thus,
	\begin{align}
	\no \mathbb{P}\left({\widetilde{\textbf{p}_{\Pi(1)}}} \in \mathcal{B}^{(n)}\right) &\leq 1- (r-1)\left(2e^{-\frac{1}{3}\left(cn^{\frac{2}{s(r-1)}+\alpha}\right)\left(\frac{1}{n^{\frac{1}{s(r-1)}+\frac{\alpha}{4}}}- \frac{c'}{n^{\frac{1}{s(r-1)} + \beta}}\right)^2}\right) \to 1,\ \
	\end{align}
as $n \to \infty$.
	
	Now, we need to show that as $n$ goes to infinity,
	\[\mathbb{P}\left( \bigcup\limits_{u=2}^{s} \left\{{\widetilde{\textbf{p}_{\Pi(u)}}}\in \mathcal{B}^{(n)}\right\}\right) \to 0.\]
	First, we show as $n$ goes to infinity,
	\[\mathbb{P}\left(\bigcup\limits_{u=2}^s \left\{\textbf{p}_u \in \mathcal{C'}^{(n)}\right\} \right)\to 0. \]
	Note for all $u \in \{2, 3, \cdots,s\}$,
	\[4 \left(\Delta_n\right)^{r-1} \delta_1 < \mathbb{P}\left( \textbf{p}_u\in \mathcal{C'}^{(n)}\right) < 4 \left(\Delta_n\right)^{r-1} \delta_2,\]
	and according to the union bound,
	\begin{align}
	\no \mathbb{P}\left( \bigcup\limits_{u=2}^s \left\{\textbf{p}_u \in \mathcal{C}^{(n)}\right\} \right) &\leq \sum\limits_{u=2}^s \mathbb{P}\left( \textbf{p}_u\in \mathcal{C}^{(n)}\right) \\
	\nonumber &\leq 4s \left(\Delta_n\right)^{r-1} \delta_2\\
	\nonumber &\leq 4s \frac{1}{n^{\frac{1}{s}+{\frac{\alpha(r-1)}{4}}}} \delta_2 \to 0; \ \
	\end{align}
as $n \to \infty$. Thus, all $\textbf{p}_u$'s are outside of $\mathcal{C}^{(n)}$ with high probability.
	
	Now, we claim that given all $\textbf{p}_u$'s are outside of $\mathcal{C}^{(n)}$, $\mathbb{P}\left({\widetilde{\textbf{p}_{\Pi (u)}}} \in \mathcal{B}^{(n)}\right)$ is arbitrarily small.
	Note that for all $u \in \{2,3, \cdots, s\}$ and all $i \in \{1,2, \cdots, r-1\}$, (\ref{obfs3}) and the union bounds yield
	\begin{align}
	\no \mathbb{P}\left({\widetilde{\textbf{p}_{\Pi(u)}}}\in \mathcal{B}^{(n)}\right) &\leq \mathbb{P}\left(\abs*{{\widetilde{\textbf{p}_{\Pi(u)}}}-\textbf{p}_u } \geq\Delta_n\right) \\
	\no & \leq \sum\limits_{i=1}^{r-1}\mathbb{P}\left(\abs*{{\widetilde{{p}_{\Pi(u)}(i)}}-{p}_u(i) } \geq\Delta_n\right)\\
	\no & \leq 2(r-1)e^{-\frac{1}{3}\left(cn^{\frac{2}{s(r-1)}+\alpha}\right)\left(\frac{1}{n^{\frac{1}{s(r-1)}+\frac{\alpha}{4}}}- \frac{c'}{n^{\frac{1}{s(r-1)} + \beta}}\right)^2}.\ \
	\end{align}
	As a result, by using a union bound again, as $n$ becomes large,
	\begin{align}
	\no \mathbb{P}\left( \bigcup\limits_{u=2}^s \left\{\abs*{\widetilde{\textbf{p}_{\Pi(u)}}-\textbf{p}_u}\geq \Delta_n\right\}\right)\leq s\left(2(r-1)e^{-\frac{1}{3}\left(cn^{\frac{2}{s(r-1)}+\alpha}\right)\left(\frac{1}{n^{\frac{1}{s(r-1)}+\frac{\alpha}{4}}}- \frac{c'}{n^{\frac{1}{s(r-1)} + \beta}}\right)^2}\right) \to 0.
	\end{align}
	Thus, for all $u\in \{2, 3, \cdots, s\}$, $\widetilde{\textbf{p}_{\Pi(u)}}$'s are close to $\textbf{p}_u$'s, thus they will be outside of $\mathcal{B}^{(n)}$. Now, we can conclude as $n \to \infty$ that:
	
	\begin{align}
	\no \mathbb{P}\left( \bigcup\limits_{u=2}^s \left\{{\widetilde{\textbf{p}_{\Pi(u)}}}\in \mathcal{B}^{(n)}\right\}\right) \to 0.
	\end{align}
So the adversary can successfully recover $Z_1(k)$. Since $Z_{1}(k)=X_1(k)$ with probability $1-R_u=1-o(1)$, the adversary can recover $X_{1}(k)$ with vanishing error probability.		
\end{proof}

}

}

\subsection{{\color{black}$r$-State} Markov Chain Model}
\label{markov_thm_converse_obfs}
In this section, users' data patterns are modeled using Markov chains and there are $r$ possibilities for users' data patterns. Similar to Section~\ref{markov}, we assume $p_u(i)$'s are drawn independently from some continuous density function, $f_\textbf{P}(\textbf{p}_u)$, on the $(0,1)^{|E|-r}$ hypercube, and $\textbf{p}_u$, $f_\textbf{P}(\textbf{p}_u)$, and $\mathcal{R}_\textbf{P}$ are defined in Section~\ref{markov}.

By using the general idea stated in Section~\ref{markov}, we can now repeat the similar reasoning as Theorem~\ref{r_state_thm_converse_obfs} to show the following theorem.

\begin{thm}\label{markov_thm_obfs}
	For an irreducible, aperiodic Markov chain model, iff ${\textbf{Z}}$ is the obfuscated version of $\textbf{X}$, and ${\textbf{Y}}$ is the anonymized version of ${\textbf{Z}}$ as defined above, the size of the group including user $1$ is $s$, and
	\begin{itemize}
		\item {\color{black}$m =\Omega\left(n^{\frac{2}{s(|E|-r)} + \alpha}\right)$} for any $\alpha>0$;
		\item $R_u \sim \Unif [0, a_n]$, where {\color{black}$a_n = O\left(n^{-\frac{1}{s(|E|-r)}-\beta}\right)$} for any $\beta>\frac{\alpha}{4}$;
	\end{itemize}
then, user $1$ has no privacy at time $k$.
\end{thm}

{\color{black}
	\section{More General Setting for the Association Graph}
	\label{appendixA}
The association graph structure that we have studied so far was somewhat general except for one aspect: We assumed that people in a group can have dependency but they are independent from members of other groups. It is natural to assume that there could be dependency between members of each group and outside members. Here we discuss how to apply the developed results to this more general setting. 

Similar to~\cite{Com1, Com2, Com3,Com4, Com5, Com6,Com7,Com8, Com_Negar, Com_Negar2}, we consider a community structure with strong intra-community connections and weak inter-community connection. In the community structure the nodes of the network can be grouped into sets of users such that each set of users is densely connected internally as shown in Figure~\ref{fig:com1}. Here, we also assume that the adversary has some knowledge about all of the covariances between users in addition to the marginal probability distributions: they know whether the value of each covariance is less than or higher than a specific threshold. We show that the adversary can reliably reconstruct the entire association graph for \textit{the anonymized version of the data} (i.e., the observed data traces) with relatively few observations.
	
	\begin{figure*}[t!]
		\centering
		\subfloat[A sketch of a small network displaying community structure.]{
			\includegraphics[width=0.5\columnwidth]{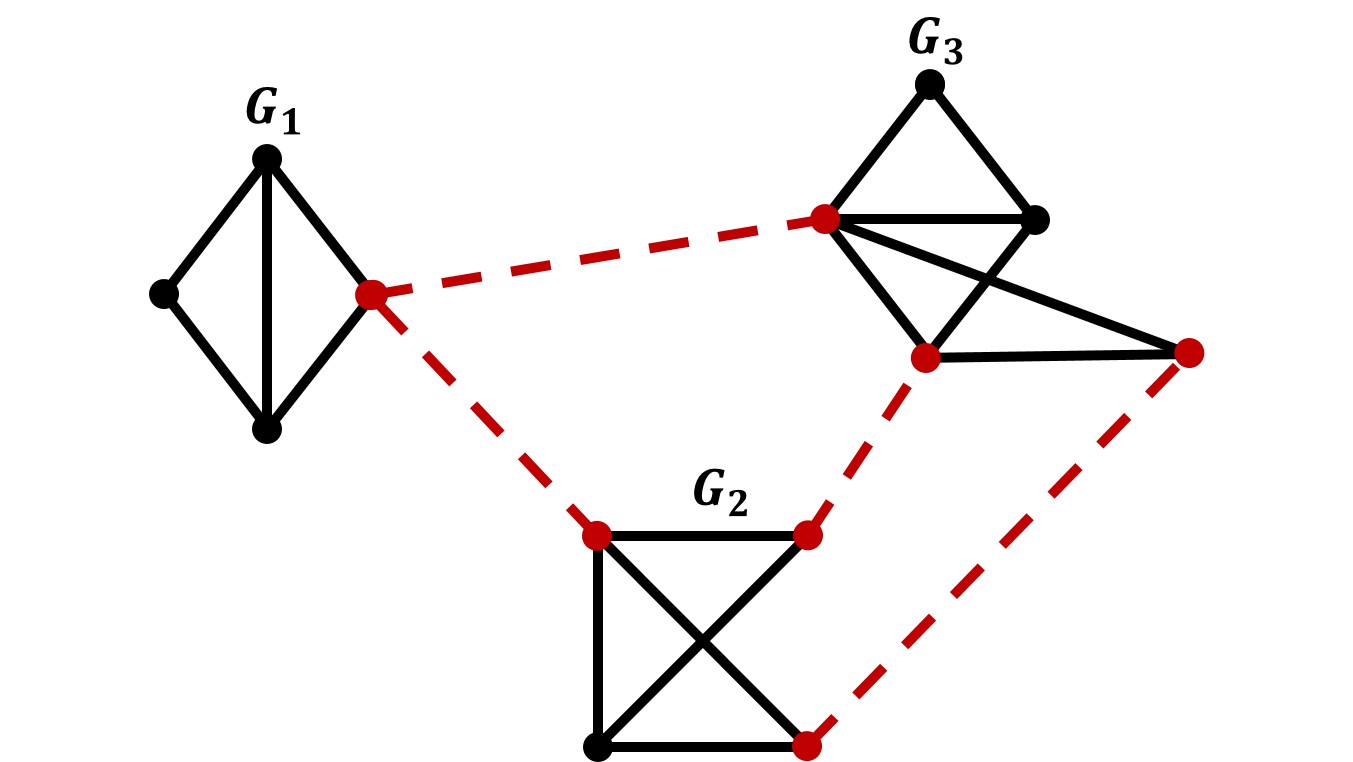}
			\label{fig:com1}
		}
		\subfloat[The association graph consists of disjoint subgraphs.]{
			\includegraphics[width=0.5\columnwidth]{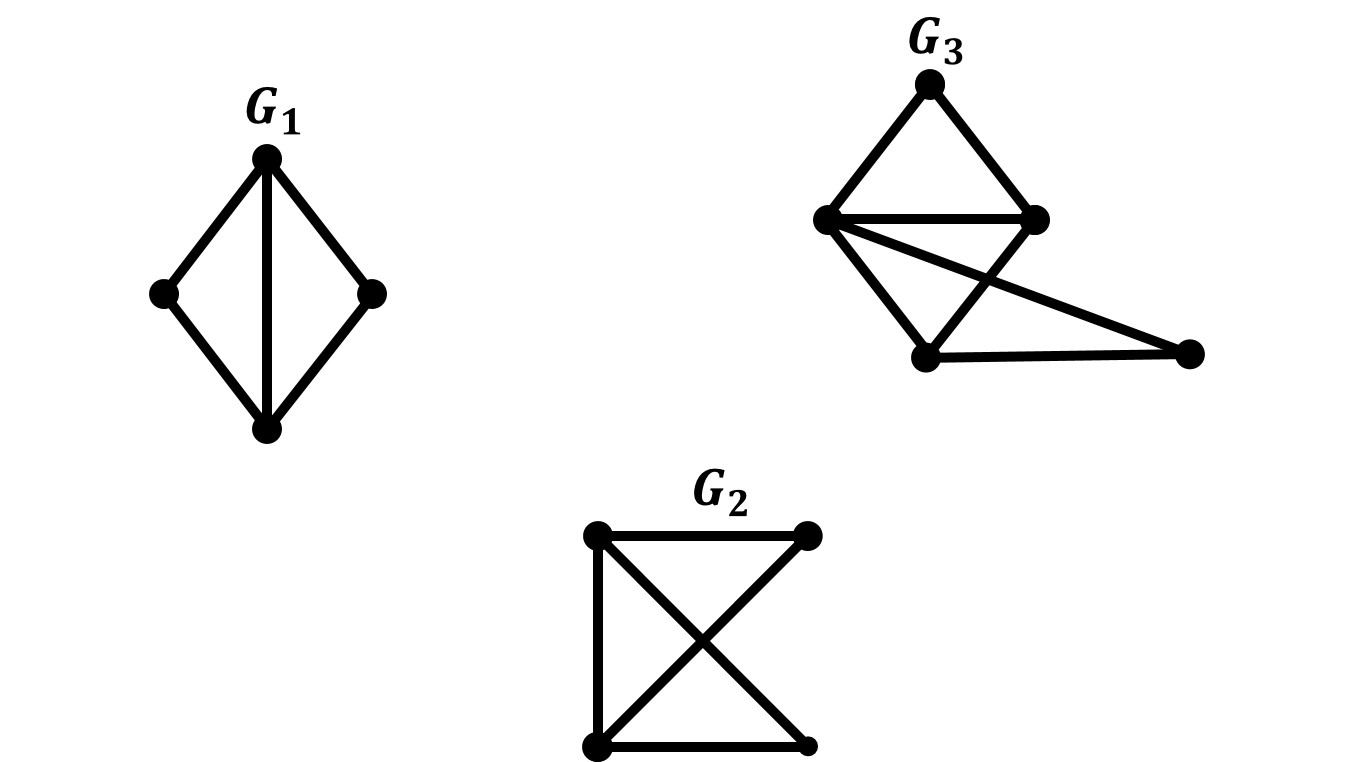}
			\label{fig:com2}
		}
		\caption{The adversary uses their prior knowledge to break inter-community edges.}
		\label{fig:com33}
	\end{figure*}
	
	Let $G(\mathcal{V},F)$ denote the association graph with set of nodes $\mathcal{V}$, $(|\mathcal{V}|=n)$, and set of edges $F$. In this case, we use an association graph based on a threshold as follows: we assume two vertices (users) are connected if their data sets are strongly correlated, and are not connected if their data sets are weakly correlated. More specifically,
\begin{itemize}
	\item $(u,u') \notin F$ iff $\cov\left(X_u(k); X_{u'}(k)\right)\leq\epsilon_1,$
	\item $(u,u') \in F$ iff $\cov\left(X_u(k);X_{u'}(k)\right)\geq \epsilon_2,$
\end{itemize}
where $\cov\left(X_u(k); X_{u'}(k)\right)$ is the covariance between the $k^{th}$ data point of user $u$ and user $u'$.

\begin{lem}
	\label{lem3}
		Consider a general association graph, $G(\mathcal{V},F)$, based on the threshold as described above. If the adversary obtains $m=(\log n)^3$ anonymized observations per user, they can construct $\widetilde{G}=\widetilde{G}(\widetilde{\mathcal{V}}, \widetilde{F})$, where $\widetilde{\mathcal{V}}=\{\Pi(u):u \in \mathcal{V}\}=\mathcal{V}$, such that with high probability, for all $u, u' \in \mathcal{V}$; $ (u,u')\in F$ iff $\left(\Pi(u),\Pi(u')\right)\in \widetilde{F}$. We write this statement as $\mathbb{P}(\widetilde{G}\simeq G)\to 1.$
\end{lem}
	
	\begin{proof}

		Note for $u, u' \in \{1, 2, \cdots, n\}$, we write $v=\Pi(u)$ and $v'=\Pi(u')$. We provide an algorithm for the adversary that with high probability obtains all edges of $F$ correctly. For each pair $w$ and $w'$, the adversary computes $\widetilde{Cov_{vv'}}$ as follows:
		\begin{align}
		\widetilde{Cov_{vv'}}&=\frac{\sum\limits_{i=1}^{r-1}\sum\limits_{j=1}^{r-1}ij\widehat{M}_{vv'}(i,j)}{m}-\frac{\sum\limits_{i=1}^{r-1}i\widehat{M}_{v}(i)}{m}\frac{\sum\limits_{i=1}^{r-1}i\widehat{M}_{v'}(i)}{m}\ \
		\label{cov}
		\end{align}
		where
		\[\widehat{M}_{vv'}(i,j)= \abs*{\left\{k: Y_v(k)=i, Y_{v'}(k)=j\right\}}.\]
		\[\widehat{M}_{v}(i)= \abs*{\left\{k: Y_v(k)=i\right\}}.\]
		\[\widehat{M}_{v'}(j)= \abs*{\left\{k: Y_{v'}(k)=j\right\}}.\]
		After observing $m=\left(\log n\right)^3$ data points per user and computing the above expressions, the adversary constructs $\widetilde{G}$ in the following way:
		\begin{itemize}
			\item If $|\widetilde{Cov_{vv'}}|\leq\epsilon_1$, then $(v,v')\notin \widetilde{F}.$
			\item If $|\widetilde{Cov_{vv'}}|\geq\epsilon_2$, then $(v,v')\in \widetilde{F}.$
		\end{itemize}
		We show the above method yields $\mathbb{P}(\widetilde{G}\simeq G)\to 1$ as $n \to \infty$, as follows. Note
		\[\widehat{M}_{vv'}(i,j) \sim { \Bino } (m, w_{vv'}(i,j)),\] \[\widehat{M}_{v}(i) \sim { \Bino } (m, w_{v}(i)),\] \[M_{v'}(i) \sim { \Bino } (m, w_{v'}(i)),\]
		where $w_{vv'}(i,j)=\mathbb{P}\left(Y_v(k)=i, Y_{v'}(k)=j\right)$, $w_{v}(i)=\mathbb{P}\left(Y_v(k)=i\right)$, and $w_{v'}(i)=\mathbb{P}\left(Y_{v'}(k)=i\right)$. From proof of Lemma~\ref{lem1}, by using (\ref{eq_1_1}), for all $v,v' \in \{1,2,\cdots, n\}$ and all $i, j \in \{0, 1,\cdots, r-1\}$, we have
	\begin{align}
	0\leq mw_{vv'}(i,j)-m^{\frac{3}{4}} \leq \widehat{M}_{vv'}(i,j) \leq mw_{vv'}(i,j)+m^{\frac{3}{4}}.\ \
	\label{M4}
	\end{align}
	\begin{align}
	0\leq mw_v(i)-m^{\frac{3}{4}} \leq \widehat{M}_v(i) \leq mw_v(i)+m^{\frac{3}{4}}.\ \
	\label{M5}
	\end{align}
	\begin{align}
	0\leq mw_{v'}(i)-m^{\frac{3}{4}} \leq \widehat{M}_{v'}(i) \leq mw_{v}(i)+m^{\frac{3}{4}}.\ \
	\label{M6}
	\end{align}
$A_{vv'}(i,j)$ is defined as the event that (\ref{M4}), (\ref{M5}), and (\ref{M6}) are all valid, thus, based on proof of Lemma~\ref{lem1}, we have
\begin{align}
\mathbb{P}\left(\bigcap\limits_{v=1}^n \bigcap\limits_{v'=1}^n\bigcap\limits_{i=0}^{r-1} \bigcap\limits_{j=0}^{r-1}\left\{A_{vv'}(i,j)\right\}\right) \to 1,\ \
\label{A2}
\end{align}
as $n \to \infty$. Let us define $ C_{vv'
}=\bigcap\limits_{i=1}^{r-1}\bigcap\limits_{j=1}^{r-1}\left\{A_{vv'}(i,j)\right\}$. Now, if $C_{vv'}$ is true for some $v,v' \in \{1,2,\cdots, n\}$, according to (\ref{cov}), we have
	\begin{align}
		\no \widetilde{Cov_{vv'}}&=\frac{\sum\limits_{i=1}^{r-1}\sum\limits_{j=1}^{r-1}ij\widehat{M}_{vv'}(i,j)}{m}-\frac{\sum\limits_{i=1}^{r-1}i\widehat{M}_{v}(i)}{m}\frac{\sum\limits_{i=1}^{r-1}i\widehat{M}_{v'}(i)}{m}\\
		\no &\leq\frac{\sum\limits_{i=1}^{r-1}\sum\limits_{j=1}^{r-1}ij\left(mw_{vv'}(i,j)+m^{\frac{3}{4}}\right)}{m}-\frac{\sum\limits_{i=1}^{r-1}i\left(mw_{v}(i)-m^{\frac{3}{4}}\right)}{m}\frac{\sum\limits_{i=1}^{r-1}i\left(mw_{v'}(i)-m^{\frac{3}{4}}\right)}{m}\\
		\no&= \sum\limits_{i=1}^{r-1}\sum\limits_{j=1}^{r-1}ij w_{vv'}(i,j)-\sum\limits_{i=1}^{r-1}iw_v(i)\sum\limits_{i=1}^{r-1}iw_{v'}(i)\\
		\no & \hspace{0.5 in}+\frac{r^2(r-1)^2}{4}m^{-\frac{1}{4}}+\frac{r(r-1)}{2} \sum\limits_{i=1}^{r-1}i(w_v(i)+w_{v'}(i))m^{-\frac{1}{4}}-\frac{r^2(r-1)^2}{4}m^{-\frac{1}{2}}\\
		&\leq Cov_{vv'}+\frac{r^2(r-1)^2}{4}m^{-\frac{1}{4}}+\frac{r(r-1)}{2} \sum\limits_{i=1}^{r-1}i(w_v(i)+w_{v'}(i))m^{-\frac{1}{4}}+\frac{r^2(r-1)^2}{4}m^{-\frac{1}{2}},\ \
		\label{eq11}
		\end{align}
		where $Cov_{vv'}=\sum\limits_{i=1}^{r-1}\sum\limits_{j=1}^{r-1}ij w_{vv'}(i,j)-\sum\limits_{i=1}^{r-1}iw_v(i)\sum\limits_{i=1}^{r-1}iw_{v'}(i)$.
		Similarly,
		\begin{align}
		\no \widetilde{Cov_{vv'}}&=\frac{\sum\limits_{i=1}^{r-1}\sum\limits_{j=1}^{r-1}ij\widehat{M}_{vv'}(i,j)}{m}-\frac{\sum\limits_{i=1}^{r-1}i\widehat{M}{w}(i)}{m}\frac{\sum\limits_{i=1}^{r-1}i\widehat{M}_{v'}(i)}{m}\\
		\no &\geq\frac{\sum\limits_{i=1}^{r-1}i\left(mw_{v}(i)-m^{\frac{3}{4}}\right)}{m}-\frac{\sum\limits_{i=1}^{r-1}i\left(mw_{v}(i)+m^{\frac{3}{4}}\right)}{m}\frac{\sum\limits_{i=1}^{r-1}i\left(mw_{v'}(i)+m^{\frac{3}{4}}\right)}{m}\\
		&= Cov_{vv'}-\frac{r^2(r-1)^2}{4}m^{-\frac{1}{4}}-\frac{r(r-1)}{2} \sum\limits_{i=1}^{r-1}i(w_v(i)+w_{v'}(i))m^{-\frac{1}{4}}-\frac{r^2(r-1)^2}{4}m^{-\frac{1}{2}}.\ \
		\label{eq12}
		\end{align}
		Now, by using (\ref{eq11}) and (\ref{eq12}), we have
		\begin{align}
		\abs*{\widetilde{Cov_{vv'}}-Cov_{vv'}} \leq \frac{r^2(r-1)^2}{4}m^{-\frac{1}{4}}+\frac{r(r-1)}{2} \sum\limits_{i=1}^{r-1}i(w_v(i)+w_{v'}(i))m^{-\frac{1}{4}}+\frac{r^2(r-1)^2}{4}m^{-\frac{1}{2}}.\ \
		\label{cov1}
		\end{align}
Let us define event $D_{vv'}$ as the event that (\ref{cov1}) is valid, thus, we have shown, for all $v,v' \in \{1,2,\cdots, n\}$, $ C_{vv'} \subseteq D_{vv'}$, and consequently,
$$\left\{\bigcap\limits_{v=1}^n \bigcap\limits_{v'=1}^n\left\{ C_{vv'}\right\} \right\} \subseteq \left\{\bigcap\limits_{v=1}^n \bigcap\limits_{v'=1}^n\left\{D_{vv'}\right\} \right\}. $$
As a result,
\begin{align}
\no \mathbb{P}\left(\bigcap\limits_{v=1}^n \bigcap\limits_{v'=1}^n\left\{D_{vv'}\right\}\right) \geq \mathbb{P}\left(\bigcap\limits_{v=1}^n \bigcap\limits_{v'=1}^n\left\{C_{vv'}\right\}\right).\ \
\end{align}
Thus, by using (\ref{A2}), we have 
\begin{align}
\no \mathbb{P}\left(\bigcap\limits_{v=1}^n \bigcap\limits_{v'=1}^n\left\{D_{vv'}\right\}\right) \to 1,\ \
\end{align}
as $n \to \infty$. Hence, with high probability, for all $v,v' \in \{1,2,\cdots, n\}$ , we have
	
	\begin{align}
		\abs*{\widetilde{Cov_{vv'}}-Cov_{vv'}} \leq \frac{r^2(r-1)^2}{4}m^{-\frac{1}{4}}+\frac{r(r-1)}{2} \sum\limits_{i=1}^{r-1}i(w_v(i)+w_{v'}(i))m^{-\frac{1}{4}}+\frac{r^2(r-1)^2}{4}m^{-\frac{1}{2}}.\ \
	\end{align}
Thus, we can conclude, with high probability, for all $v,v' \in\{1, 2, \cdots, n\}$, $\widetilde{Cov_{vv'}}$'s are close to $Cov_{vv'}$'s.

Now, if $(u,u')$ is an inter-community edge, the adversary knows $Cov_{uu'} \leq \epsilon_1$, and as a result, $Cov_{vv'} \leq \epsilon_1$, thus, the adversary removes that edge. Now, we can conclude $(v,v') \notin \widetilde{F}$, and in other words, $(\Pi(u), \Pi(u')) \notin \widetilde{F}.$
This is true with high probability, simultaneously for all $u, u' \in \{1, 2, \cdots, n\}$ where $(u,u')$ is an inter-community edge.

 In addition, if $(u,u')$ is an intra-community edge, the adversary knows $Cov_{uu'} \geq \epsilon_2$, and as a result, $Cov_{vv'} \geq \epsilon_2$. Now, we can conclude $(v,v') \in \widetilde{F}$, and in other words, $(\Pi(u), \Pi(u')) \in \widetilde{F}.$ This is true with high probability, simultaneously for all $u, u' \in \{1, 2, \cdots, n\}$ where $(u,u')$ is an intra-community edge.

		As a result, for large enough $n$, we have $\mathbb{P}\left(\widetilde{G}\simeq G\right)\to 1$, so the adversary can reconstruct the association graph of the anonymized version of the data which is based on a threshold with an arbitrarily small error probability.	
	\end{proof}
	Now, the adversary has a graph structure shown in Figure~\ref{fig:com2}, where subgraph $G_1$ is a connected graph with $s_1$ vertices which is disjoint from the reminder of the association graph $(G'=G-G_1)$. In other words,
	\[G=G_1 \cup G'.\]

Now, we can repeat the same reasoning as that in the proof of Theorem~\ref{r_state_thm}, Theorem~\ref{markov_thm}, Theorem~\ref{r_state_thm_converse_obfs}, and Theorem~\ref{markov_thm_obfs} to obtain the same results for this case.
	
\hspace{-0.18 in}\textbf{{Discussion $7$}:}
The stochastic block model is a generative model for random graphs~\cite{block1,block2,block3,block4,block5, block6,block7}. Note that there are two key differences between the stochastic block model and the work here. First, in the stochastic block model, the edge set is sampled at random and the probability distributions of edges are the key part of the work, while here the analysis is based on the users' data traces, and the statistical knowledge of the adversary is a key part. Second, in the stochastic block model, nodes within a community connect to nodes in other communities in an equivalent way. In other words, any two vertices $u\in C_{i}$ and $v\in C_{j}$ are connected by an edge with probability $p_{ij}$, where $C_i$ and $C_j$ are different blocks, so all edges between two communities have the same weights or strengths. While here, as shown in Figure~\ref{fig:com1}, there is no need that the inter-community edges, corresponding to the covariance of nodes in separate communities, has the same value as others; in other words, there is no need for the nodes in a community to connect to the nodes in other communities in an equivalent way. In our work, for each of intra-community edges, $Cov_{uu'} \geq \epsilon_2$, and for each of inter-community edges, $Cov_{uu'} \leq \epsilon_1$, thus, edges have different weights.
}

%
%

\section{{\color{black}Improving} Privacy in the Presence of Dependency }
\label{perfect}
In the previous parts of this paper, we argued and demonstrated that inter-user dependency degrades the privacy provided by standard privacy-preserving mechanisms (PPMs). In this section, we discuss how to design PPMs considering inter-user dependency in order to better preserve privacy.
First, note that independent obfuscation alone cannot be sufficient even at a high noise level, because it cannot change the association graph. Therefore, the adversary can still reconstruct the association graph with a small number of observations if we add independent obfuscation noise. To mitigate this issue, we suggest that associated users collaborate in applying the noise when deploying a PPM.

For clarity, we focus on the two-state i.i.d.\ case ($r=2$). In the first part, we also focus on the case the association graph consists of subgraphs with the size of each of them less than or equal to $2$ $(s_l \leq 2)$. Thus, according to Figure~\ref{fig:simpleexam}, there are some connected users and there are also some isolated users. First, we state the following lemma.

\begin{figure}[h]
	\centering
	\includegraphics[width = 0.5\linewidth]{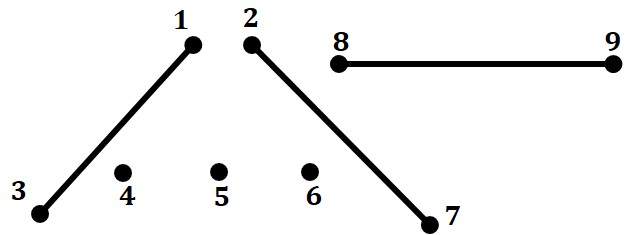}
	\caption{Graph $G$ consists of some subgraphs ($G_l$) with $s_l \leq 2$.}
	\label{fig:simpleexam}
\end{figure}

\begin{lem}
	\label{lem2}
	Let $X_u(k) \sim \Bern (p_u)$ and $X_{u'}(k) \sim \Bern (p_{u'})$; then, there exists an obfuscation technique with an asymptotic noise level
	\[\check{a}(u,{u'})= \frac{{\cov(X_u(k), X_{u'}(k))}}{\max \{p_u, p_{u'}, 1-p_u, 1-p_{u'}\}},\]
	for the dataset of user $u$ and user ${u'}$ such that $\check{Z}_u(k)$ and $\check{Z}_{u'}(k)$ are independent from each other. Note $\check{Z}_u(k)$ and $\check{Z}_{u'}(k)$ are the $k^{th}$ (reported) data point of user $u$ and ${u'}$, respectively, after applying obfuscation with the noise level equal to $\check{a}(u,{u'})$.
\end{lem}

\begin{proof}
	We explain the idea behind this lemma by an example.
	\begin{example}
		Let $X_u(k) \sim \Bern \left(\frac{3}{5}\right)$ and $X_{u'}(k) \sim \Bern \left(\frac{1}{5}\right)$, and let Table~\ref{tab:mathfunc} show the joint probability mass function of $X_u(k)$ and $X_{u'}(k)$.
		
		\begin{table}[ht]
 			 \centering
			\caption{Joint probability mass function of $X_u(k)$ and $X_{u'}(k)$.}
	 \begin{tabular}{l|c|c}
					\backslashbox{$X_u(k)$}{$X_{u'}(k)$} & $0$ & $1$ \\
					\hline
					$0$ & $\frac{7}{20}$& $\frac{1}{20}$ \\	\hline
					$1$ & $\frac{9}{20}$ & $\frac{3}{20}$ \\
				\end{tabular}
			\label{tab:mathfunc}
		\end{table}
		As a result, if we observe $2000$ bits of data, Table~\ref{tab:databits} shows the expected results according to Table~\ref{tab:mathfunc}.
		
		Then, to make $\check{Z}_u(k)$ and $\check{Z}_{u'}(k)$ independent, it is sufficient for $\check{Z}_u(k)|\check{Z}_{u'}(k) =0$ to have the same distribution as $\check{Z}_u(k)|\check{Z}_{u'}(k)=1 $. This means we should have
		\begin{align}
			\no \frac{\mathbb{P}\left(\check{Z}_u(k)=1,\check{Z}_{u'}(k)=1\right)}{\mathbb{P}\left(\check{Z}_{u'}(k)=1\right)}=\frac{\mathbb{P}\left(\check{Z}_u(k)=1,\check{Z}_{u'}(k)=0\right)}{\mathbb{P}\left(\check{Z}_{u'}(k)=0\right)};\ \
		\end{align}
thus, according to Table~\ref{tab:databits2},
		\[\frac{300 (1-\Upsilon)}{100+300}=\frac{900}{700+900} \to \Upsilon=\frac{1}{4},\]
		where $\Upsilon$ is the portion of data points $X_{u}(k)$ that need to be flipped in the fourth region of Table~\ref{tab:databits} (i.e., the region $X_u(k)=1, X_{u'}(k)=1$).
Now, we need to change $\frac{1}{4}\cdot 300=75$ of data bits. As a result, if $X_u(k)=1$ and $X_{u'}(k)=1$, then, we pass $X_u(k)$ through a $BSC(\frac{1}{4})$, and obtain $\check{Z}_u(k)$. Hence, asymptotic noise level is equal to
		\[\check{a}(u,{u'})= \frac{3}{20}\cdot\frac{1}{4}=3.75\%.\]
	\begin{table}[ht]
		 \small
 			 \centering
			\caption{(a) The expected results of $X_u(k)$ and $X_{u'}(k)$ according to Table~\ref{tab:mathfunc} after observing $2000$ bits of data, and (b) The desired results to make $\check{Z}_u(k)$ and $\check{Z}_{u'}(k)$ independent from each other.\\}
			\vspace{-1cm}
			 \subfloat[]{%
			 \hspace{2cm}%
				\begin{tabular}{l|c|c}
					\backslashbox{$X_u(k)$}{$X_{u'}(k)$} & $0$ & $1$ \\
					\hline
					$0$ & $700$& $100$ \\	\hline
					$1$ & $900$ & $300$
					\label{tab:databits}
				\end{tabular}
				 \hspace{2cm}%
			}
						 \subfloat[]{%
			 \hspace{2cm}%
					\begin{tabular}{l|c|c}
					\backslashbox{$\check{Z}_u(k)$}{$\check{Z}_{u'}(k)$} & $0$ & $1$ \\
					\hline
					$0$ & $700$& $175$ \\	\hline
					$1$ & $900$ & $225$
					\label{tab:databits2}
				\end{tabular}
				 \hspace{2cm}%
			}
		\end{table}

It is easy to check that the asymptotic noise level will be given by the equation in Lemma~\ref{lem2}. Specifically, for the above example,
		\[\cov(X_u(k), X_{u'}(k))=\mathbb{P}(X_u(k)=1, X_{u'}(k)=1)-\mathbb{P}(X_u(k)=1)\mathbb{P}(X_{u'}(k)=1)=\frac{3}{20}-\frac{3}{5}\cdot\frac{1}{5}=\frac{3}{100},\]
and we have
		\[\check{a}(u,{u'})=\frac{\cov(X_u(k), X_{u'}(k))}{\max \{p_u, p_{u'}, 1-p_u, 1-p_{u'}\}} = \frac{\frac{3}{100}}{\frac{4}{5}}=3.75\%.\]
	\end{example}
	
	Now, to prove the lemma, apply the above procedure to a general table for a probability mass function. Let $X_u(k) \sim \Bern \left(p_u\right)$ and $X_{u'}(k) \sim \Bern \left(p_{u'}\right)$. Then, to make these two sequences independent, it suffices if $\check{Z}_u(k)|\check{Z}_{u'}(k)=0$ has the same distributions as $\check{Z}_u(k)|\check{Z}_{u'}(k)=1 $. We only prove for the case $\max \{p_u,p_{u'},1-p_u,1-p_{u'}\}=1-p_{u'}$, and the proofs of the other cases are similar to this one. Now, If $X_u(k)=1$ and $X_{u'}(k)=1$, we pass $X_u(k)$ through a $BSC(\Upsilon)$ in order to obtain $\check{Z}_u(k)$. Thus,
	\[\frac{\mathbb{P}\left(X_u(k)=1,X_{u'}(k)=0\right) }{1-p_{u'}}=\frac{\mathbb{P}\left(X_u(k)=1,X_{u'}(k)=1\right)(1-\Upsilon)}{p_{u'}},\]
	and $\Upsilon$ can be calculated as
	\[\Upsilon=1-\frac{p_{u'}}{1-p_{u'}}\frac{\mathbb{P}\left(X_u(k)=1,X_{u'}(k)=0\right)}{\mathbb{P}\left(X_u(k)=1,X_{u'}(k)=1\right)}.\]
Now, we can conclude,
\begin{align}
\no \check{a}(u,{u'})&=\Upsilon \mathbb{P}\left(X_u(k)=1,X_{u'}(k)=1\right)\\
\no &=\left(1-\frac{p_{u'}}{1-p_{u'}}\frac{\mathbb{P}\left(X_u(k)=1,X_{u'}(k)=0\right)}{\mathbb{P}\left(X_u(k)=1,X_{u'}(k)=1\right)}\right) \mathbb{P}\left(X_u(k)=1,X_{u'}(k)=1\right)\\
\no &= \mathbb{P}\left(X_u(k)=1,X_{u'}(k)=1\right)-\frac{p_{u'}}{1-p_{u'}}{\mathbb{P}\left(X_u(k)=1,X_{u'}(k)=0\right)}\\
\no &= \frac{\mathbb{P}\left(X_u(k)=1,X_{u'}(k)=1\right) -p_{u'}\left(\mathbb{P}\left(X_u(k)=1,X_{u'}(k)=1\right)+\mathbb{P}\left(X_u(k)=1,X_{u'}(k)=0\right)\right)}{1-p_{u'}}\\
\no &=\frac{\mathbb{P}\left(X_u(k)=1,X_{u'}(k)=1\right)-p_up_{u'}}{1-p_{u'}}\\
\no &=\frac{\cov(X_u(k), X_{u'}(k))}{\max\{p_u,p_{u'},1-p_u,1-p_{u'}\}}. \ \
\end{align}

\end{proof}

Lemma~\ref{lem2} provides a method to convert correlated data to independent traces. The remaining task is to show that we can achieve perfect privacy after applying such a method. As shown in Figure~\ref{fig:xyzz}, two stages of obfuscation and one stage of anonymization are employed to achieve perfect privacy for users. Note that the first stage of obfuscation is due to Lemma~\ref{lem2} and the second stage (as will be explained in the proof) is the same obfuscation technique given in Theorem 1 of~\cite{Nazanin_IT}. In Figure~\ref{fig:xyzz}, $\check{Z}_u(k)$ shows the (reported) data point of user $u$ at time $k$ after applying the first stage of obfuscation with the noise level equal to
\[\check{a}(u,u')=\frac{\cov(X_u(k), X_{u'}(k))}{\max\{p_u,p_{u'},1-p_u,1-p_{u'}\}}\]
for the dataset of user $u$ and user ${u'}$, ${Z}_u(k)$ shows the (reported) data point of user $u$ at time $k$ after applying the second stage of obfuscation with the noise level equal to
\[a_n=c'n^{-\left(\frac{1}{s}-\beta\right)},\]
and $Y_u(k)$ shows the (reported) data point of user $u$ at time $k$ after applying anonymization.
\begin{figure}[h]
	\centering
	\includegraphics[width = \linewidth]{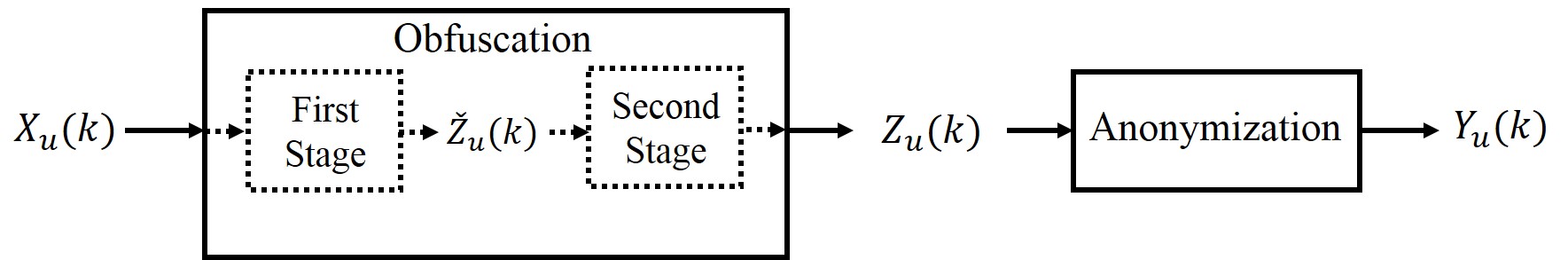}
	\caption{Applying obfuscation and anonymization techniques to the users' data points.}
	\label{fig:xyzz}
\end{figure}

 Consider $G(\mathcal{V},F)$, where $s_l \leq 2$. We have the same model for $p_u$ as in the previous sections: $p_u$ is chosen from some density $f_P(p_u)$ such that, for $\delta_1, \delta_2>0$:
\begin{equation}
\no\begin{cases}
\delta_1<f_P(p_u) <\delta_2, & p_u \in (0,1).\\
f_P(p_u)=0, & p _u\notin (0,1).
\end{cases}
\end{equation}

Also, if $(u,u') \in F$, $\rho_{uu'}$ is chosen according to some density $f_{P} (\rho_{uu'} | p_u,p_{u'})$ with range of $\left[0, \min \left\{\sqrt{\frac{p_u(1-p_{u'})}{p_{u'}(1-p_{u})}}, \sqrt{\frac{p_{u'}(1-p_u)}{p_u(1-p_{u'})}}\right\}\right].$ The following theorem states that we can indeed achieve perfect privacy if we allow collaboration between users.

\begin{thm}
	\label{thm:noise}
	For the two-state model, if ${\textbf{Z}}$ is the obfuscated version of $\textbf{X}$, ${\textbf{Y}}$ is the anonymized version of ${\textbf{Z}}$, and the size of all subgraphs are less than or equal to $2$, there exists an anonymization/obfuscation scheme such that for all $(u,u') \in F$, the asymptotic noise level for users $u$ and $u'$ is at most
	\[a(u,u')= \frac{{\cov (X_u(k), X_{u'}(k))}}{\max \{p_u, p_{u'}, 1-p_u, 1-p_{u'}\}},\]
	to achieve perfect privacy for all users. The anonymization parameter $m=m(n)$ can be made arbitrarily large.
\end{thm}

\begin{proof}
	There are two main steps.
	
	Step $1$: De-correlate based on Lemma~\ref{lem2}. In particular, note that for at least half of the users, no noise is added in this step. More specifically, define
	\[\mathcal{U}= \text{Set of unaffected users} =\{u: \text{no noise is added to user } u \text{ in this step}\}.\]
Then after step $1$, we have $\check{Z}_u(k) \sim \Bern (\check{q}_u)$. As a result,
	\begin{itemize}
		\item For $u \in \mathcal{U}$; $\check{Z}_u(k)=X_u(k)$ and $\check{q}_u=p_u.$
		\item For $u \in \{1, 2, \cdots, n\}-\mathcal{U}$; $\check{Z}_u(k)\neq X_u(k)$ and $\check{q}_u\neq p_u.$
	\end{itemize}
	Note $|\mathcal{U}| \geq \frac{n}{2}$, because the main graph consists of some subgraphs with $s_l\leq 2$.
	
	Step $2$: Assume $\check{q}_u$'s are known to the adversary. The setup is now very similar to Theorem $1$ in~\cite{Nazanin_IT}, where perfect privacy is proved for the i.i.d.\ data. But there is a difference here. Specifically, although the users' data $\check{Z}_u(k)$ are now independent, the distribution of $\check{q}_u$'s are not, since they are the result of the data-dependent obfuscation technique of Lemma~\ref{lem2}. Luckily, this issue can be easily resolved so that we can show perfect privacy for user $1$. The main idea is to use the fact that as stated above, at least $\frac{n}{2}$ of the users are not impacted by the de-correlation step. As we see below, these users will be sufficient to ensure perfect privacy for user $1$ (which may or may not be in the set $\mathcal{U}$).

Let's explore the distributions of $\check{Q}_u=\check{q}_u$ for users in the set $\mathcal{U}$. For any correlated pair of users, the method of Lemma~\ref{lem2} leaves the one whose $p_u$ is farthest from $\frac{1}{2}$ intact. Since $p_u$'s are chosen independently from each other and each user is correlated with only one user, it is easy to see that for users in the set $\mathcal{U}$, the $\check{q}_u$'s are i.i.d.\ with the following probability density function
\[f_{\check{Q}}(\check{q}_u)=2 f_P(\check{q}_u) \int_{\min(\check{q}_u,1-\check{q}_u)}^{\max(\check{q}_u,1-\check{q}_u)} f_P(x) dx. \]
Therefore, the setup is the same as Theorem $1$ in~\cite{Nazanin_IT} where we want to prove perfect privacy for user $1$, and we have $\frac{n}{2}$ users who are independent from user $1$ and their parameter $\check{q}_u$ is chosen i.i.d.\ according to a density function. However, we need to check that the density function $f_{\check{Q}}(\check{q})$ satisfies the condition $\check{\delta}_1 <f_{\check{Q}}(\check{q}
_u)<\check{\delta}_2$ for some $\check{\delta}_1$ and $\check{\delta}_2$ on a neighborhood $\check{q}_u \in [p_u-\epsilon',p_u+\epsilon']$.
First, note that
\begin{align*}
 f_{\check{Q}}(\check{q}_u) &=2 f_P(\check{q}_u) \int_{\min(\check{q}_u,1-\check{q}_u)}^{\max(\check{q}_u,1-\check{q}_u)} f_P(x) dx. \\
 & < 2 \delta_2^2=\check{\delta}_2.
\end{align*}
Next,
\begin{align*}
 f_{\check{Q}}(\check{q}_u) &=2 f_P(\check{q}_u) \int_{\min(\check{q}_u,1-\check{q}_u)}^{\max(\check{q}_u,1-\check{q}_u)} f_P(x) dx. \\
 & > 2 \delta_1^2 |1-2\check{q}_u|=\check{\delta}_1.
\end{align*}
Thus, as long as $p_u \neq \frac{1}{2}$, the condition is satisfied \footnote{The case $p_u = \frac{1}{2}$ has zero probability, and thus need not be considered. Nevertheless, the result can be shown for $p_u=\frac{1}{2}$, as all we require is a number of users proportional to the length of the interval in the vicinity of $p_u$.}. Therefore, we can show perfect privacy for user $1$. Note that here, in the second step, we need to apply a second stage of obfuscation and apply anonymization according to Theorem $1$ in~\cite{Nazanin_IT}. Nevertheless, since the noise level $a_n \to 0$ for this second stage, the asymptotic noise level will stay the same as that for step $1$, i.e.
\[a(u,u')= \check{a}(u,u')=\frac{\abs{\cov (X_u(k), X_{u'}(k))}}{\max \{p_u, p_{u'}, 1-p_u, 1-p_{u'}\}}.\]
\end{proof}

Now, the above method can be extended to the case where $s_l>2$. Let $s_l=3$. From Figure~\ref{Sj=3}, there are two different situations in this case:
\begin{enumerate}

	\item Case $1$: As shown in Figure~\ref{Sj_case2}, user $1$ and user $2$ are correlated, and user $2$ and user $3$ are correlated. In the first step, we de-correlate user $2$ and user $3$ based on Lemma~\ref{lem2}. {\color{black}Now, we face a similar situation as that in the case $s_l=2$ (as shown in Figure~\ref{Sj_case1}), and we de-correlate them based on Lemma~\ref{lem2}.} Hence, we can make all of the users independent from each other and then according to Theorem~\ref{thm:noise}, we can achieve perfect privacy for all of them.
	\item Case $2$: As shown in Figure~\ref{Sj_case3}, all three users are correlated to each other. In the first step, we use Lemma~\ref{lem2} to make user $1$ and user $3$ uncorrelated. Now, we have a similar situation as case $1$, so we can make all the users independent from each other and then, according to Theorem~\ref{thm:noise}, we can achieve perfect privacy for all of them.
\end{enumerate}
\begin{figure*}[t!]
	\centering
	\subfloat[$s_l=2$.]{
		\includegraphics[width=0.23\columnwidth]{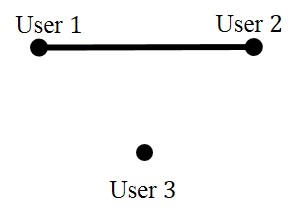}
		\label{Sj_case1}
	}
	\subfloat[$s_l=3$: Case $1$.]{
		\includegraphics[width=0.23\columnwidth]{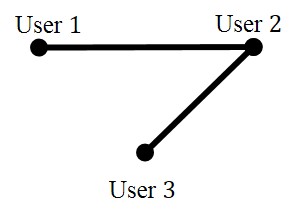}
		\label{Sj_case2}
	}
	\subfloat[$s_l=3$: Case $2$.]{
	\includegraphics[width=0.23\columnwidth]{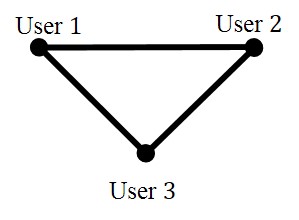}
	\label{Sj_case3}
}
	\caption{Three different ways which $3$ users can be correlated to each other.}
	\label{Sj=3}
\end{figure*}

\hspace{-0.18 in}\textbf{{Discussion $8$}:} Note that obfuscating data by adding non-zero asymptotic noise may degrade utility significantly. Therefore, in practice, it is usually not possible to de-correlate \emph{all} dependent users without imposing substantial utility degradation. In addition, in order to convert correlated data to independent data, users should collaborate together and disclose their private data to each other, which degrades privacy unless users trust each other. In such a setting, a possible approach in applying our technique is to only add {\color{black}de-correlation} noise to the data of highly-dependent users (e.g., spouses and close friends), and leave data of less-dependent users (e.g., co-workers) unchanged.




{\color{black}
	\begin{table}[]
		\caption{Summary of the results for the case anonymization is employed as a PPM for "no privacy" as a function of number of adversary's observations per user $(m)$. Here, $s$ is the size of group of users whose data traces are dependent, $r$ is the number of possible values for each user's data point, $|E|$ is the size of set of edges in the Markov chain, and the results hold for any $\alpha >0.$}
		{\color{black}\begin{center}
				\begin{tabular}{|c|c|c|}
					\hline
					\multirow{2}{*}{Users' data model} & {Independent users~\cite{tifs2016}} & {Dependent users} \\ \cline{2-3} 
					& $m$ & $m$ \\ \hline
					Two-state i.i.d. model & $\Omega\left(n^{{2}+\alpha}\right)$ &$\Omega\left(n^{\frac{2}{s}+\alpha}\right)$ \\ \hline
					$r$-state i.i.d. model & $\Omega\left(n^{\frac{2}{r-1}+\alpha}\right)$ & $\Omega\left(n^{\frac{2}{s(r-1)}+\alpha}\right)$ \\ \hline
					$r$-state Markov chain model& $\Omega\left(n^{\frac{2}{|E|-r}+\alpha}\right)$ &$\Omega\left(n^{\frac{2}{s(|E|-r)}+\alpha}\right)$ \\ \hline
				\end{tabular}
		\end{center}}
		\label{table1}
	\end{table}
}

{\color{black}
	\begin{table}[]
		\caption{Summary of the results for the case both obfuscation and anonymization are combined to be employed as a PPM for "no privacy" as a function of number of adversary's observations per user $(m)$ and the amount of noise level $(a_n)$. Here, $s$ is the size of group of users whose data traces are dependent, $r$ is the number of possible values for each user's data point, $|E|$ is the size of set of edges in the Markov chain, and the results hold for any $\alpha >0$.}
		{\color{black}\begin{center}
				\begin{tabular}{|c|c|c|c|c|}
					\hline
					\multirow{2}{*}{Users' data model} & \multicolumn{2}{c|}{Independent users~\cite{Nazanin_IT}} & \multicolumn{2}{c|}{Dependent users} \\ \cline{2-5}
					& $m$ & $a_n$& $m$ & $a_n$ \\ \hline
					Two-state i.i.d. model & $\Omega\left(n^{2+\alpha}\right)$ & $O\left(n^{-1-\beta}\right)$ & $\Omega\left(n^{\frac{2}{s}+\alpha}\right)$ & $O\left(n^{-\frac{1}{s}-\beta}\right)$ \\ \hline
					$r$-state i.i.d. model & $\Omega\left(n^{\frac{2}{r-1}+\alpha}\right)$ & $O\left(n^{-\frac{1}{r-1}-\beta}\right)$ & $\Omega\left(n^{\frac{2}{s(r-1)}+\alpha}\right)$ & $O\left(n^{-\frac{1}{s(r-1)}-\beta}\right)$ \\ \hline
					$r$-state Markov chain model & $\Omega\left(n^{\frac{2}{|E|-r}+\alpha}\right)$ & $O\left(n^{-\frac{1}{|E|-r}-\beta}\right)$ & $\Omega\left(n^{\frac{2}{s(|E|-r)}+\alpha}\right)$ & $O\left(n^{-\frac{1}{s(|E|-r)}-\beta}\right)$ \\ \hline
				\end{tabular}
		\end{center}}
		\label{table2}
	\end{table}
	
}
{\color{black}
\section{Acknowledgment}
\label{Ackno}
The authors would like to thank Dr. Farhad Shirani Chaharsooghi (the New York University) for his valuable suggestions and
discussions regarding graph alignments.}

\section{Conclusion}
\label{conclusion}
Resourceful adversaries can leverage statistical matching based on the prior behavior of users in order to break the privacy provided by PPMs.
Our previous work has considered the requirements on anonymization and obfuscation for ``perfect'' user privacy when traces are independent between users. However, in practice users have correlated data traces, as relationships between users establish dependence in their behavior. In this paper, we demonstrated that such dependency degrades the privacy of PPMs, as the anonymization employed must be significantly increased to preserve perfect privacy, and often no degree of independent obfuscation of the traces can be effective. {\color{black}The summary of the results is shown in Tables~\ref{table1} and~\ref{table2}.} We have also presented preliminary results on dependent obfuscation to {\color{black}improve users' privacy}.



\bibliographystyle{IEEEtran}
\bibliography{REF}
%
%
%
%

 \end{document}